\documentclass[11pt]{article}
\def\colorful{1}
\usepackage{quantum}
\usepackage{graphicx}

\newcommand{\ZZ}{\Z}
\newcommand{\FF}{{\F_2}}
\newcommand{\id}{{\mathrm{I}}}

\newcommand{\cF}{{\cal F}}
\newcommand{\cB}{{\cal B}}
\newcommand{\cE}{{\cal E}}
\newcommand{\cI}{{\cal I}}
\newcommand{\cC}{{\cal C}}

\newcommand{\nB}{{n_{\calB}}}

\newcommand{\mB}{{m_{\calB}}}
\newcommand{\nF}{{n_{\calF}}}
\newcommand{\mF}{{m_{\calF}}}
\newcommand{\cA}{{\cal A}}

\newcommand{\cobd}{{\partial^\transp}}

\newcommand{\sanH}{{\scalebox{.75}[1.0]{\textsc{Horz}}}}

\newcommand{\XConst}{10^{5}}

\begin{document}

\title{Fiber Bundle Codes: \\ Breaking the $N^{1/2} \polylog(N)$ Barrier for Quantum LDPC Codes}

\author{Matthew B. Hastings\thanks{Station Q. and Microsoft Quantum.}
\and
Jeongwan Haah\thanks{Station Q. and Microsoft Quantum.}
\and
Ryan O'Donnell%
\thanks{Microsoft Quantum and Carnegie Mellon University Computer Science Department.}
}

\date{\today}
\maketitle

\begin{abstract}
We present a quantum LDPC code family that 
has distance $\Omega(N^{3/5}/\polylog(N))$ and $\tilde\Theta(N^{3/5})$ logical qubits, 
where $N$ is the code length.
This is the first quantum LDPC code construction which achieves distance greater than $N^{1/2} \polylog(N)$.
The construction is based on generalizing the homological product of codes to a fiber bundle.
\end{abstract}

\maketitle
\section{Introduction}
While there are many constructions of ``good" classical LDPC codes with linear rate and distance~\cite{RU08}, the construction of a quantum LDPC code\footnote{Throughout, by ``quantum code", we mean a CSS stabilizer code.  An LDPC (low-density parity check) should have stabilizer generators of weight $O(1)$, and each qubit should be in the support of at most $O(1)$ stabilizer generators.} on $N$ qubits with distance greater $N^{c}$ for some $c > \frac 1 2 $ has been a longstanding open problem.
Even constructing a code with distance $\Omega(N^{1/2}\polylog(N))$ is nontrivial, and many natural constructions such as the toric code~\cite{Kit03} give only distance $\Theta(N^{1/2})$.
The first code to beat $N^{1/2}$ by a polylogarithm was based on the cellulation of a carefully chosen manifold~\cite{FML02}.
A later construction based on Bruhat--Tits buildings improved the polylogarithm and gave an efficient decoder~\cite{EKZ20}, but still had distance only $O(N^{1/2} \polylog(N))$, with
later generalizations~\cite{KT20} allowing an arbitrarily large power in the polylog.
In this paper, we give a construction attaining distance $\Omega(N^{3/5}/\polylog(N))$, solving this problem.  We present partial results toward efficient decoding, giving a polynomial-time algorithm to decode bit flip errors and a conjectured efficient algorithm to decode phase errors.

One of the difficulties in constructing a good quantum code is that the $X$- and $Z$-stabilizer generators must commute with each other.
If one set of generators --- say the $X$-generators --- is chosen randomly (trying to follow randomized constructions of classical codes), it is unlikely that there will exists a set of low weight $Z$~generators that commute with them.
So, while these randomized constructions are useful if one allows high-weight generators, it is necessary to incorporate some structure into the code if one wishes to have low-weight generators.
The approach in~\cite{EKZ20} uses deep algebraic/number-theoretic structure.
In this paper we follow a different approach, a fiber bundle construction, that combines a simple fiber (the cycle graph) with a random base (a classical LDPC code).

A precursor to this construction is the homological product of quantum codes~\cite{FH14,BH14}.
This product takes two quantum codes and constructs a product code from them.
This product was used~\cite{BH14} to construct quantum codes with linear distance and rate, and with generators of weight $O(N^{1/2})$, by taking the product of two random quantum codes with linear-weight generators.
This product has several other applications.
The hypergraph product~\cite{TZ14,LTZ15} is a particular form of the homological product when both codes are classical, giving quantum LDPC codes with linear rate and distance~$\Theta(N^{1/2})$.
Another application is to distance balancing, taking a quantum code which has different distances $d_X$ and $d_Z$ against $X$- and $Z$-errors, and increasing the number of qubits to increase one of the distances: so long as $\sqrt{d_X d_Z} \gg N^{1/2}$, this gives~\cite{Has17qic,EKZ20} a new quantum code with both $d_X, d_Z \gg N^{1/2}$.

We would like to contrast quantum LDPC codes, which are a subclass of stabilizer codes,
with more general quantum subsystem codes~\cite{Poulin2005}.
Both are defined by a set of Pauli ``check'' operators (gauge operators),
which are individually measured when one implements the codes.
The code distance in either setting is defined 
as the minimum of weights of operators that act nontrivially on encoded qubits.
The difference between stabilizer and subsystem codes is that 
the subsystem code does not impose the condition that these check operators should commute with one another
and hence in general there is no interpretation of a subsystem code as a chain complex.
It is natural to ask our question in a relaxed setting of subsystem codes,
and there is a construction under this relaxation:
there exists a ``sparse'' subsystem code family where all check operators have a constant weight
and where the code distance is $\Omega(N^{1-\epsilon})$ with $\epsilon = O((\log N)^{-1/2})$~\cite{BFHS2014}.
This family violates our requirement that all check operators must commute.
Let alone the mathematical difference,
we believe there is an intrinsic merit to quantum LDPC stabilizer codes over sparse subsystem codes.
Indeed, implementing a quantum error correcting code, be it subsystem or stabilizer, 
we have to infer the syndrome, the eigenvalues of stabilizers.
For quantum LDPC codes, a check operator is a stabilizer 
so the syndrome measurement of a stabilizer involves only a constant number of data qubits.
However, in a subsystem code, even if check operators have constant weight,
a stabilizer may be a product of a large number of check operators 
and it is not always clear if the syndrome can be read off in a fault-tolerant manner.

\subsection{Results, motivation, and outline}
In this paper, we generalize the homological product to a \emph{twisted homological product} based on the idea of fiber bundles from topology, giving what we call \emph{fiber bundle codes}.
The main result is the following theorem for codes with \emph{polylogarithmic weight} stabilizers and with each qubit participating in polylogarithmically many stabilizers\footnote{In this paper, for simplicity of presentation we do not optimize polylogarithms in the distance.  There are some simple ways in which the polylogarithm in our main theorem can be improved by changing some of our parameter choices later, and we comment on them where appropriate.}.
Then, by applying weight reduction and distance balancing techniques we obtain LDPC codes given in \Cref{coro}.

\begin{theorem} \label{mainth}
    There exists a family of quantum codes on $N$ qubits with  $d_X=\Omega(N^{1/2}/\polylog(N))$ and $d_Z=\Omega(N^{3/4}/\polylog(N))$ where all stabilizer generators have weight at most $\polylog(N)$ and all qubits participate in at most $\polylog(N)$ stabilizer generators.
    The code has $\Theta(N^{1/2})$ logical qubits.
\end{theorem}

This code is not LDPC, but in  \Cref{wrc}, we show how to weight reduce it to an LDPC code at only a polylogarithmic cost in distance and number of physical qubits (the construction in the Appendix may be of more general interest also, as we use a notion of homotopy equivalence between chain complexes to relate different quantum codes).  The presentation in the Appendix is self-contained, and is based on weight reducing a certain classical code used in the bundle construction\footnote{A previous version of the present paper appealed to the general weight reduction result of~\cite{Has17qic} but G. Z\'{e}mor has pointed out an error in that paper.  A corrected version of weight reduction for arbitrary quantum codes will appear separately.}.

The distances of this LDPC code are not balanced since $d_Z \gg d_X$, but as mentioned we can apply the distance balancing technique~\cite{Has17qic,EKZ20}.
The distance balancing procedure of~\cite{EKZ20} generalizes that of~\cite{Has17qic} and improves the rate of the resulting code.  
This gives us the following:
\begin{corollary} \label{coro}
    There exists a family of quantum LDPC codes on $N$ qubits having distance \linebreak \mbox{$d=\Omega(N^{3/5}/\polylog(N))$} and with $\Omega(N^{3/5}/\polylog(N))$ logical qubits.
\end{corollary}

There is a long history~\cite{Kit03,FM01} of applying ideas from topology to quantum codes, since a quantum code can naturally be interpreted as a chain complex and such a complex can be derived from cellulations of a manifold.
Then, operations which have a natural definition in terms of manifolds can often be translated into useful operations on quantum codes.
For example, the product of two manifolds naturally leads to considering the homological product of quantum codes.

A fiber bundle is a generalization of the idea of taking a product of two manifolds; roughly, it is something that locally looks like a product but has richer global structure.
A simple example of a fiber bundle is a M\"{o}bius strip: locally it ``looks like" the product of a circle with an interval.
However, there is a global twist: going once around the circle reverses the interval.
This makes the M\"{o}bius band not homeomorphic to the product of a circle with an interval.
More generally, we consider a \emph{base} (a circle, in the M\"obius band)
and a \emph{fiber}, where the fiber admits automorphisms (reversing the interval, in the M\"obius band).
Motion along the base can involve acting on the fiber by some automorphism.

Another simple example is the case where both the base and the fiber are a circle,~$S^1$.
Their untwisted, usual product is a torus.
One can impose a twist so that the fiber is reflected when going around the base circle (for example, using angular coordinates $\phi$ for the fiber, one maps $\phi \mapsto -\phi$); this changes the topology to that of a Klein bottle.

Beyond the changes in topology, these twists can also have an interesting effect on the geometry.
Consider again the example where both the base and the fiber are circles.
Rather than imposing reflection, we can impose a rotation by some fixed angle $\phi_0 \in \R$ when going around the base circle.
This does not change the topology, so the result is still a torus;
it is parameterized by angles $\theta \in \R$ for the base and $\phi \in \R$ for the fiber,
and we identify
\[
    (\theta,\phi) \equiv (\theta, \phi + 2\pi) \equiv (\theta + 2\pi, \phi + \phi_0).
\]
For any value of $\phi_0$, this is still a torus, but the geometry is different.
In the context of quantum codes, one needs some cellulation of the manifold, so one may cellulate the fiber and base circles in the obvious way (by cycle graphs $C_{n_B}$, $C_{n_F}$ for some integers $n_B,n_F>0$),
where the twist $\phi_0$ of the fiber is an integer multiple of $2\pi / n_F$.
In this case, the result is still a toric code but with different geometry.

Interestingly, even in this very simple case, the change in geometry resulting from this twist can improve the distance of the code!
Taking no twist ($\phi_0 = 0$), the code has $N=2n_B n_F$ qubits and a distance equal to $\min(n_B, n_F)$,
so that the distance is equal to $\sqrt{N/2}$ when $n_F = n_B$.
We leave it to the reader to work out the details, but by imposing an appropriate twist and changing $n_B$ and $n_F$, one may construct a code whose distance is $\sqrt{cN}$ for some $c>1/2$.
While this does not improve the scaling of the code with~$N$, it is a quantitative improvement in distance.
Generalizing this construction to higher dimensions~\cite{Has17itcs}, and assuming an unproven conjecture in geometry, this could allow for the construction of LDPC codes with distance $N^{1-\epsilon}$ for any $\epsilon>0$.

In this paper, we consider a further such idea, combining these twists with the use of randomness.
We will take a very simple choice of fiber (a circle), but we will choose the base to be a random LDPC code, considered as a chain complex.
While the algebraic ideas will be familiar for those with a topology background, we would like warn these readers that many of our choices of chain complexes do not have a nice interpretation as cellulations of a manifold.
For example, the base of our bundle will be a ``$1$''-complex
\[
    \calB_1 \xrightarrow{\partial} \calB_0
\]
with polylogarithmically many ``$0$''-cells in the boundary of each ``$1$''-cell, while usual cellulations of a manifold (or any topological cell complex) would have by definition only two $0$-cells in the boundary of a $1$-cell.
Even more strangely, the base will have zeroth Betti number equal to zero, $b_0=0$, while of course usually the zeroth Betti number is the number of connected components of the manifold.
As explained in~\cite{BH14}, it is possible to ``reverse engineer'' a manifold of high dimension from the code constructed here, in which our ``$1$''-complex will no longer represent the $1$-dimensional skeleton of a high dimensional manifold.
One can also reverse engineer a $3$-complex from the quantum LDPC code constructed here, and triangulate it with simplices to get a simplicial $3$-complex.

The paper is outlined as follows.  In \Cref{reviewsection}, we review the connection between quantum codes and cohomology.
In \Cref{definesection} we define the fiber bundle code.  In this section, we pick a specific choice of fiber, and pick the base to be a classical code but we leave the construction of the classical code for later.
Much of this section defines bundles in general and computes (co)homology of bundles, and then
in \Cref{sketchsection}, we define the fiber bundle code and sketch the main results needed to prove \Cref{mainth}.
Then in \Cref{randbase} we give the randomized construction of the base code and prove lower bounds on the weight of cohomology and homology representatives.  The homology representative bound depends on some complicated properties of an associated classical code proven in \Cref{codingboundsection}; this section has the most detailed combinatorial calculations.
In \Cref{decodingsection}, we present partial results toward an efficient decoding algorithm.
Finally, \Cref{notationsection} collects some of the notation that we use.

\subsection{Quantum codes, chain complexes, and (co)homology}
\label{reviewsection}

In this work all quantum codes are CSS quantum codes on qubits.
All vector spaces will be over $\F_2$ and all homology and cohomology takes coefficients in~$\F_2$.

Let us briefly review notions of homological algebra.
A \emph{chain complex}
\[
    \cdots \xrightarrow{\partial_{j+1}} \calA_j \xrightarrow{\partial_j} \calA_{j-1} \xrightarrow{\partial_{j-1}} \cdots \xrightarrow{\partial_1} \calA_0 \xrightarrow{ \bdry_0 = 0 } 0
\]
is a sequence of vector spaces $\calA_0, \calA_1,\dots$, each with some preferred basis, together with linear maps $\partial_j : \calA_j \to \calA_{j-1}$ called \emph{boundary operators} between these vector spaces.
The boundary maps obey the condition $\partial_j \partial_{j+1} = 0$ whenever $\partial_j$ and $\partial_{j+1}$ are defined.
A \emph{$k$-complex} is a chain complex with $\calA_j = 0$ for all $j > k$ but $\calA_k \neq 0$.
We refer to basis elements of ${\cal A}_j$ as \emph{$j$-cells}, and to vectors in ${\cal A}_j$ as \emph{$j$-chains}.
So, a cell is a particular chain.
The \emph{Hamming weight} of a chain is the number of cells in the chain with a nonzero coefficient; we write the Hamming weight  using absolute value symbols $|\ldots|$.
We sometimes identify a chain with the set of cells that have nonzero coefficient in the chain;
since the coefficient field is $\F_2$ this identification does not forget any data of a chain.

The \emph{homology} of a chain complex $\calA$ is a sequence of vector spaces
\[
H_j(\calA) = \ker \bdry_j / \im \bdry_{j+1}
\]
where $j = 0,1,\ldots$. The \emph{$j$-th Betti number} is defined%
\footnote{
Usually, Betti numbers are defined as the rank of the free part of the homology with integer coefficients.
Our definition is different from this usual one when the integral homology has $2$-torsions.
}
as
\[
b_j(\calA) = \dim_\FF H_j(\calA).
\]
It is customary to assume $\bdry_{k+1} = 0$ for a $k$-complex
even if $\bdry_{k+1}$ is not explicitly mentioned.
An element of $\ker \bdry_j$ is called a \emph{$j$-cycle}.
The \emph{cohomology} $H^j(\calA)$ is the homology of its dual chain
\begin{align*}
&\left( \cdots \xleftarrow{\bdry^*_{j+1}} \calA_j^* \xleftarrow{\bdry^*_{j}} \calA_{j-1}^* \xleftarrow{\bdry^*_{j-1}} \cdots
\xleftarrow{\bdry^*_1} \calA_0^* \xleftarrow{0} 0 \right)\\
\cong&
\left(
\cdots \xleftarrow{\bdry^\transp_{j+1}} \calA_j \xleftarrow{\bdry^\transp_{j}} \calA_{j-1} \xleftarrow{\bdry^\transp_{j-1}} \cdots
\xleftarrow{\bdry^\transp_1} \calA_0 \xleftarrow{0} 0
\right)
\end{align*}
where $\calA_j^*$ is the vector space of linear functionals on $\calA_j$.
Many authors write the \emph{coboundary maps}~$\bdry_j^*$ as ``$\delta_{j-1}$,'' but we will not.
The indicated isomorphism is nothing but a collection of isomorphisms $\calA_j^* \cong \calA_j$,
which are established thanks to the preferred basis for each $\calA_j$.
More concretely, a linear functional $c^* \in \calA_j^*$
that assigns $1 \in \FF$ for a cell $c \in \calA_j$
but $0 \in \FF$ for all other cells,
is identified  with the cell $c$ itself under the isomorphism.
Almost always in the context of cohomology,
an element of $\calA_j^*$ is called a ``cochain'';
however,
in this paper we just call it a \emph{chain} since we always use the isomorphism $\calA_j^* \cong \calA_j$.
That is, a chain will always be a $\FF$-linear combination of cells
regardless of whether we use the chain for homology or cohomology.
The isomorphism $\calA_j^* \cong \calA_j$ is fundamental to any combinatorics of cohomology:
we will have to count the weight of a cohomology representative, called a \emph{cocycle},
via this identification of linear functionals with their cell-support.

A chain complex defines a quantum code by picking some integer $q > 0$ and associating $q$-cells of the complex with qubits,
and associating $(q-1)$- and $(q+1)$-cells with $X$- and $Z$-stabilizer generators of the code (respectively).
The $Z$ logical operators of the code are associated with $q$th homology classes and the $X$ logical operators are associated with $q$th cohomology classes.
The code has two distances, denoted $d_X$ and $d_Z$, where $d_X$ is the weight of a lowest weight nontrivial $X$ logical operator (i.e., the lowest possible Hamming weight of a vector that represents nontrivial $q$th cohomology) and $d_Z$ is the weight of a lowest weight nontrivial $Z$ logical operator (i.e., the lowest possible Hamming weight of a vector that represents nontrivial $q$th homology).
Note that the $\FF$-dimensions of $H_q(\calA)$ and $H^q(\calA)$ are always the same
as seen by counting vector space dimensions and matrix ranks.

\section{Fiber Bundle Codes}
\label{definesection}

In this section we define the code.
As remarked earlier, the code results from a chain complex,
and thus we focus on constructing chain complexes.
We proceed from a general possible construction to our specific instantiation.
We begin with reviewing the product of two chain complexes in \Cref{sec:chain-product},
and recall the (untwisted) homological product in \Cref{sec:trivialbundle}.
We next explain in \Cref{sec:twistbundle} that general fiber bundles are obtained by twisting boundary maps.
Deferring specific choices of base and fiber complexes and twists,
we study algebraic aspects of fiber bundle complexes
to establish homology and cohomology isomorphisms at dimension~1 in \Cref{sec:H1iso}.
We finally define our code in \Cref{sec:circle-bundle} by specifying the fiber complex;
the base complex will be a random classical code with polylogarithmic weight parity checks,
whose combinatorial properties will be studied in the next sections.
In \Cref{cohomd,homd}
we lower bound the weight of cohomology and homology representatives for this choice, using probabilistic methods.

\subsection{Products of chain complexes}\label{sec:chain-product}

Let us begin by recalling the definition of a homological product.
Given a \emph{base} complex ${\calB}$ and a \emph{fiber} complex ${\calF}$,%
\footnote{
The untwisted, usual homological product does not distinguish between base and fiber.
}
we construct their product $\cE$, a \emph{bundle}, as follows.
We take tensor products of component chain vector spaces,
which inherit the boundary maps from the constituent complexes~$\calB$ and~$\calF$:
\begin{equation}   \label{com-diagram}
\begin{gathered}
\xymatrix{
\cdots \ar[d] &\ar[l] \cdots \ar[d]& \calB_j \otimes \calF_k \ar[d]^{\id \otimes \partial_k}\ar[l]_{\partial_j \otimes \id}\\
\calB_0 \otimes \calF_1 \ar[d]^{\id \otimes \partial_1} &\ar[l]_{\partial_1 \otimes \id} \calB_1 \otimes \calF_1 \ar[d]^{\id \otimes \partial_1}& \vdots \ar[l] \ar[d]\\
\calB_0 \otimes \calF_0 &\ar[l]_{\partial_1 \otimes \id } \calB_1 \otimes \calF_0 & \vdots \ar[l]
}
\end{gathered}
\end{equation}
Then, the chain space $\calE_r$ of the bundle is defined to be the direct sum
\[
\calE_r = \bigoplus_{p+q = r} \calE_{p,q} \quad \text{ where } \quad  \calE_{p,q} = \calB_p \otimes \calF_q
\]
along the diagonal line $p+q = r$ in the diagram for each $r \ge 0$.
A preferred basis of the chain vector space $\calE_r$
is also inherited from constituent cells;
every $r$-cell of $\calE$ is a pair $(b^p,f^q)$ with $p+q = r$
where $b^p \in \calB_p$ is a $p$-cell of the base
and $f^q \in \calF_q$ is a $q$-cell of the fiber.
We will refer to such an $r$-cell of the bundle as a $(p,q)$-cell.
If ${\calB}$ and ${\calF}$ are $\mathsf b$- and $\mathsf f$-complexes, respectively,
then $\cE$ is a $(\mathsf b + \mathsf f)$-complex.

\subsection{Trivial bundles}\label{sec:trivialbundle}

We have only defined the chain spaces $\calE_r$ above, but not yet the boundary maps $\partial^\calE$.
There are in fact many ways to define $\partial^\calE$
given the diagram in \eqref{com-diagram},
and this diversity will be realized in the discussion of twisted bundles below.
Before we show such diversity,
we recall the untwisted boundary map.
It suffices to specify how the boundary map acts on each direct summand $\calE_{p,q}$ of~$\calE_r$:
\begin{equation}
\partial_r^\calE |_{(p,q)} = \id \otimes \partial_{q}^\calF \,\,+\,\, \partial_{p}^\calB \otimes \id,
\label{eq:trivial-bdef}
\end{equation}
where $\partial^\calB$ and $\partial^\calF$ are the boundary maps of $\calB$ and $\calF$, respectively.
One may verify that $\partial^\calE_r \partial^\calE_{r+1}=0$ for all $r$;
here we use that the vector spaces are over $\FF$.%
\footnote{
With a general coefficient group, an extra sign is needed:
$\partial^\calE_{(p,q)} = (-1)^p \id \otimes \partial_{q}^\calF + \partial_{p}^\calB \otimes \id$.
}

This definition of homological product via \cref{eq:trivial-bdef} to build a trivial bundle is standard in topology,
where the chain complexes are obtained from cell decompositions of two manifolds.
The product of the chain complexes corresponds to the cellulation of the product of two manifolds.
However, the algebraic construction of trivial bundles above
does not have to come from topological spaces.
For example, it has been applied~\cite{BH14}
to input chain complexes which represent random quantum codes.

From now on we will usually drop superscripts and subscripts from boundary maps;
the meaning of $\bdry$ will be obvious from the context.

\subsection{Twisted bundles}\label{sec:twistbundle}

\begin{figure}
\includegraphics[width=\textwidth, trim={1ex 62ex 83ex 0ex}, clip]{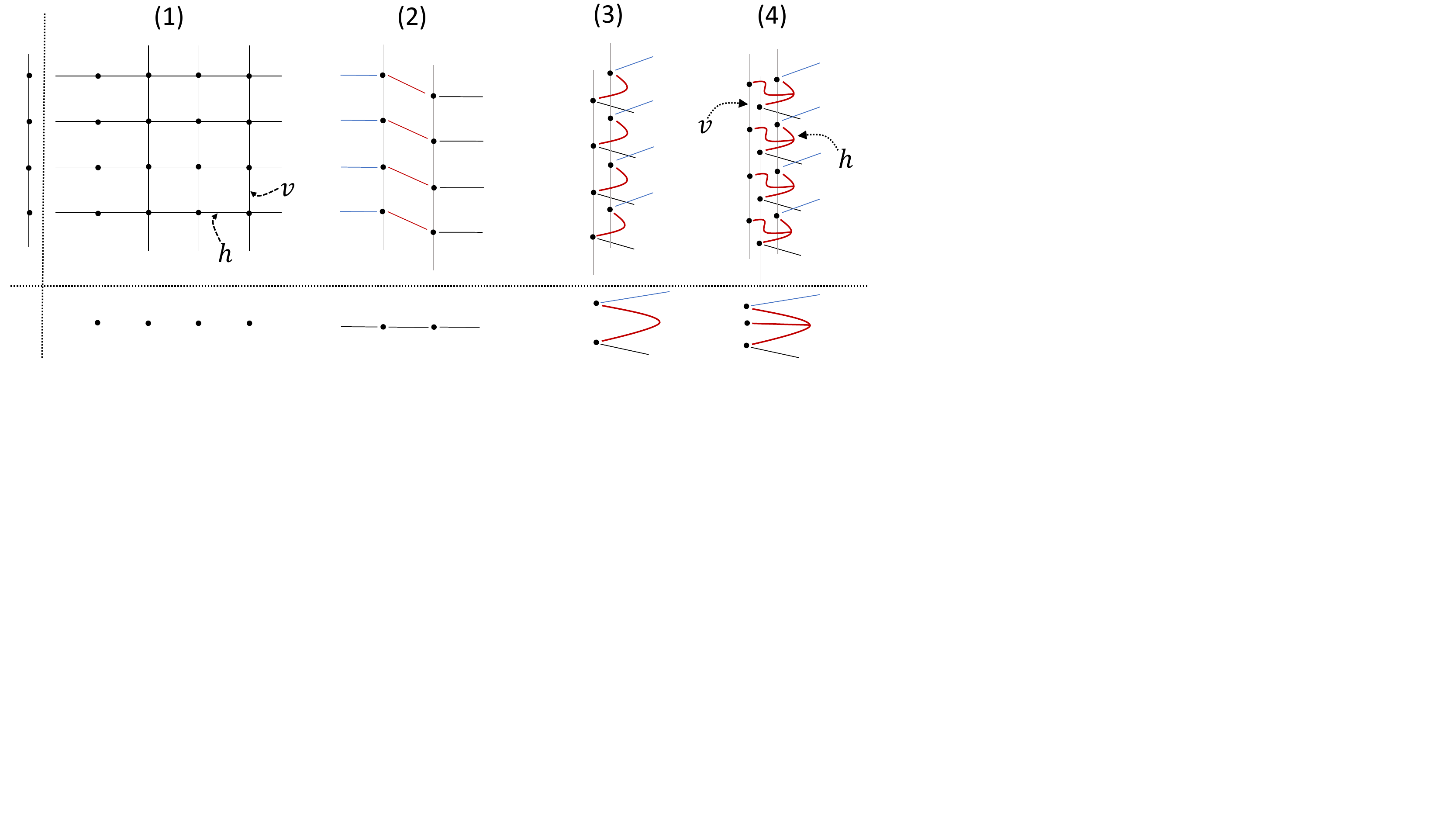}
\caption{
Fiber bundles over base $1$-complexes.
(1) depicts a trivial bundle built from two $1$-complexes,
each representing a long line or a circle.
Any bundle $1$-cell that is a lift of a base $1$-cell is referred to as a \emph{horizontal} $1$-cell,
indicated by $h$ in the figure.
Any bundle $1$-cell that vanishes upon projection onto the base is referred to as a \emph{vertical} $1$-cell,
indicated by $v$ in the figure.
The projection is defined in \Cref{def:Pi}.
(2) depicts some twisting. Since a fiber is acted on by an automorphism,
the entire fiber over a base cell is shifted.
Note that the shown twist can be removed using gauge redundancy.
To draw a nonremovable twisting, we must have had a cycle in the base.
(3) is the same as (2) but base $1$-cells are positioned on the right-hand side
and base $0$-cells on the left-hand side.
(4) introduces a ``$1$''-cell (red) of the base that has three boundary $0$-cells.
}
\label{fig:hv-cells}
\end{figure}

We assume that the fiber admits some {\em automorphism} group $G$,
which is a collection of permutation actions on the set of $q$-cells for each $q$
such that boundary operator commutes with this permutation.
Such an automorphism naturally extends by linearity to each chain vector space $\calF_q$.
The requirement of a fiber automorphism group then reads that for each~$q$,
\begin{align}
g \partial f = \partial g f \quad \text{ for all } g \in G,\,\, f^q \in \calF_q .
\label{eq:auto-consistency}
\end{align}

\begin{definition}
Given a fiber automorphism group $G$ obeying \cref{eq:auto-consistency},
a {\it connection} $\varphi$ of a bundle is an arbitrary assignment of
a automorphism group element, a {\em twist}, for each pair of a base cell and one of its boundary cell:
\begin{align}
\{ (b, a)~:~b,a\text{ are cells such that } a \in \partial b \} \xrightarrow{\quad \varphi \quad} G.
\label{eq:twist-def}
\end{align}
where we have identified $\partial b$ with its support
(the collection of cells with nonzero coefficients in $\partial b$).
We define a {\em twisted boundary map} $\partial^\calE$ by $\varphi$:%
\footnote{
With a general coefficient group, the first term of \cref{bdef} has sign $-1$.
}
\begin{align}
\partial^\calE_{(0,q)} (b^0 \otimes f) &= b^0 \otimes \partial f, \nonumber \\
\partial^\calE_{(1,q)} (b^1 \otimes f) &=
b^1 \otimes \partial f \,\,
+
\sum_{a^0 \in \partial b^1} a^0 \otimes \varphi(b^1,a^0) f.
\label{bdef}
\end{align}
\end{definition}

\begin{proposition}
If the base is a $1$-complex,
then the twisted boundary map satisfies $\partial_r^\calE \partial_{r+1}^\calE = 0$ for all $r \ge 0$.
\end{proposition}
\begin{proof}
It suffices to check the claim for basis elements.
It is obvious that
$\partial^\calE_{q-1} \partial^\calE_{q} (b^0 \otimes f^q) = b^0 \otimes \partial^\calF_{q-1} \partial^\calF_q f^q = 0$
for any $q$. If $b = b^1$ is a base $1$-cell and $f$ is a fiber $q$-cell,
by \cref{bdef} we see
\begin{align*}
\partial^\calE_q \partial^\calE_{q+1} (b^1 \otimes f^q)
&=
\partial_q^\calE(b^1 \otimes \partial f^q)
+
\partial^\calE_q \sum_{a^0 \in \partial b^1} a^0 \otimes \varphi(b^1,a^0) f^q \\
&=
\left(b^1 \otimes \partial \partial f^q
+
\sum_{a^0 \in \partial b^1} a^0 \otimes \varphi(b^1,a^0) \partial f^q \right)
+
\sum_{a^0 \in \partial b^1} a^0 \otimes \partial \varphi(b^1,a^0) f^q \\
&= 0
\end{align*}
where the second equality is because $a^0$ is a $0$-cell
and the third is because of \cref{eq:auto-consistency}.
\end{proof}

The definition of twisted boundary map can be generalized to any higher-dimensional base complex.
This generalization, however, requires certain conditions so that $\partial^\calE  \partial^\calE = 0$ is fulfilled.
For example, if the base is a $2$-complex, then we may need extra terms in the boundary map:
\begin{align}
\partial_{(2,q)}^\calE (b \otimes f) = b \otimes \partial f
+
\sum_{e \in \partial b} e \otimes \varphi(b,e) f
+
\sum_{v \in \partial e: e \in \partial b} v \otimes f^+_{v,e,b}
\end{align}
where $f^+_{v,e,b}$ is some $(q+1)$-cell of the fiber.
For such extra terms to exist,
the twists and the fiber complex should jointly obey certain conditions.
We do not pursue in this generalization further,
and from now on the base will always be a $1$-complex.
With the restriction that the base is a $1$-complex,
we will find it convenient to distinguish bundle $1$-cells as follows.
\begin{definition}
Members of $\calE_{1,0}$ are \emph{horizontal}.
Members of $\calE_{0,1}$ are \emph{vertical}.
\end{definition}
\noindent
Every bundle $1$-chain is a sum of a horizontal chain and a vertical chain,
and such a decomposition is always unique.
See \Cref{fig:hv-cells}.

While the twists described in \cref{eq:twist-def} are completely arbitrary,
not all choices of twists give different ``geometry'' for the bundle.
This is known as gauge redundancy.
Here the geometry refers to an equivalence class of bundles
where the equivalence relation is given by an isomorphism between chain complexes
such that it commutes with the twisted bundle boundary maps and sends cells to cells (preferred bases).
For example, one can transform the fiber by an automorphism $h$
and change the twists as $\varphi(b,a) \mapsto h\varphi(b,a)h^{-1}$.
Even more flexibly, one can transform a fiber over a particular base cell,
and simultaneously change the twists that connect the base cell with others.

This gauge redundancy can be so rich that
if, for example, the base is a cyclic graph (a circle),
the set of all twists can be simplified
so that there is a non-identity assignment
only for one pair of a $1$-cell and its boundary $0$-cell.
We will avoid this simplification by having a complicated base,
and this is part of reason that we will take a random code for the base.

\subsection{Isomorphisms on (co)homology groups}\label{sec:H1iso}

Under conditions we use later,
the first homology $H_1(\calE)$ and cohomology $H^1(\calE)$ of the bundle
will be isomorphic to $H_1(\calB)$ and $H^1(\calB)$ of the base, respectively.
In particular, the first Betti numbers agree: $b_1(\calE)=b_1(\calB)$.
The isomorphism will be induced by the bundle projection:
\begin{definition}\label{def:Pi}
A linear map called the {\em bundle projection} $\Pi_r: \calE_r \to \calB_r$ is defined as
\begin{align}
b^{r} \otimes f^{0} &\mapsto b^r , \nonumber\\
b^{r-j} \otimes f^{j} &\mapsto 0 \quad \text{ if } j > 0 \nonumber
\end{align}
for all $r$-cells $b^r$ and $(r-j)$-cells $b^{r-j}$ of the base,
and $0$-cells $f^0$ and $j$-cells $f^j$ of the fiber.
\end{definition}

\begin{lemma}\label{lem:H1iso}
The bundle projection induces vector space isomorphisms $\Pi_*: H_1(\calE) \to H_1(\calB)$
and $\Pi^*: H^1(\calB) \to H^1(\calE)$ if all of the following are true:
\begin{itemize}
\item[(i)] $\calB$ is a $1$-complex.
\item[(ii)] The boundary $\partial f^1$ of any fiber $1$-cell $f^1$ has even weight.
\item[(iii)] Every fiber $0$-chain of even weight is a boundary.
\item[(iv)] $H_0(\calB) = 0$, i.e.\ the zeroth Betti number $b_0(\calB)$ vanishes.
\item[(v)] Every fiber automorphism acts trivially on $H_1(\calF)$.
\end{itemize}
\end{lemma}
The conditions (ii) and (iii) are redundant in a topological setting where a $1$-cell is always a line segment.
\begin{proof}
The proof consists of the propositions below.
\begin{proposition}
Assume (i) and (ii).
Then the following diagram commutes:
\begin{equation}
\label{eq:comm-diagram}
\begin{gathered}
\xymatrix{
\calE_2 \ar[r]^{\partial}\ar[d]^{0} & \calE_1 \ar[r]^{\partial}\ar[d]^{\Pi} & \calE_0 \ar[d]^{\Pi}\\
\mathmakebox[\widthof{$\displaystyle \calB_2$}]{0=\calB_2\hphantom{{}=0}}\ar[r]_{0} & \calB_1 \ar[r]_{\partial} & \calB_0
}
\end{gathered}
\end{equation}
\end{proposition}
\begin{proof}
For the left square, we need to show that $\Pi \partial: \calE_2 \to \calB_1$ is zero.
For a base $1$-cell $b$ and fiber $1$-cell $f$,
we have $\Pi \partial (b \otimes f) = \Pi b \otimes \partial f = |\partial f| b = 0$ by (ii).
For the right square, take any $1$-chain of the bundle,
and decompose it as $h^1 + v^1$ where $h^1 \in \calE_{1,0}$ is horizontal
and $v^1 \in \calE_{0,1}$ is vertical.
It is obvious that $\Pi \partial h^1 = \partial \Pi h^1$.
By (ii), we have $\Pi \partial v^1 = 0$, and clearly $\partial \Pi v^1 = 0$.
\end{proof}

\begin{proposition}
Assume (i) and (ii).
Then the induced map $\Pi_*: H_1(\calE) \to H_1(\calB)$ is well-defined.
\end{proposition}
\begin{proof}
We have to show that
(1)~any closed $1$-chain becomes closed, and
(2)~any $1$-chain that is a boundary becomes a boundary.

Decompose a bundle $1$-chain as $h^1 + v^1$ as before.
If $h^1 + v^1$ is closed, then $\partial h^1 = \partial v^1$.
By (ii), we see $\Pi \partial v^1 = 0$.
Hence, $\partial \Pi h^1 = \Pi \partial h^1 = 0$.
This shows~(1).

For~(2), it suffices to examine boundary of $(1,1)$-cells, by (i).
We examine $\partial(b^1 \otimes f^1) = b^1 \otimes \partial f^1 + \sum_{a^0 \in \partial b^1} a^0 \otimes \varphi(b^1,a^0) f^1$.
Under $\Pi$, the sum vanishes by definition of $\Pi$
and so does the first term by (ii).
Hence, the projection of a boundary is zero.
\end{proof}

\begin{proposition}
Assume (i)---(iii).
Then $\Pi_*: H_1(\calE) \to H_1(\calB)$ is onto.
\end{proposition}
\begin{proof}
Given a homology representative cycle $b^1$ of the base,
we choose an arbitrary fiber $0$-cell $f^0$ so $b^1 \otimes f^0$ projects down to $b^1$.
Observe that
$\partial( b^1 \otimes f^0 ) = \sum_{a^0 \in \partial b^1} a^0 \otimes \varphi(b^1,a^0) f^0 = \sum_j a^0_j \otimes f'_j$
where the $a^0_j$ are distinct $0$-cells of the base and the $f'_j$ are some fiber $0$-chains.
Since $\Pi \partial (b^1 \otimes f^0) = \partial b^1 = 0$,
we have $\sum_{a^0 \in \partial b^1} \Pi (a^0 \otimes \varphi(b^1,a^0) f^0) = \sum_j |f'_j| a^0_j = 0$.
This means that $|f'_j| = 0 \bmod 2$ for all $j$.
By~(iii), we have $f'_j = \partial s^1_j$ for some fiber $1$-chain $s^1_j$.
Now, $b^1 \otimes f^0 + \sum_j a^0_j \otimes s^1_j$ is closed and projects down to $b^1$.
\end{proof}

\begin{proposition}
Assume (i)---(v).
Then $\Pi_*: H_1(\calE) \to H_1(\calB)$ is one-to-one.
\end{proposition}
\begin{proof}
Suppose $\Pi_*(v^1 + h^1) = 0$ where $v^1+h^1$ is a decomposition of a $1$-cycle of $\calE$
into vertical~($\calE_{0,1}$) and horizontal~($\calE_{1,0}$) $1$-chains.
The projection eliminates vertical $1$-cells by definition,
and takes the mod-$2$ sum of horizontal $1$-cells.
So, the vanishing projection means that over any base $1$-cell,
there are an even number of horizontal $1$-cells.
That is, we can write the horizontal chain $h^1$ as
$h^1 = \sum_j b^1_j \otimes f^0_j$ where $b^1_j$ are distinct base $1$-cells
and each $f^0_j$ is an even-weight $0$-chain of the fiber.
By (iii), $f^0_j = \partial s^1_j$ for some $1$-chain $s^1_j$ in the fiber,
and thus
$\partial \sum_j b^1_j \otimes s^1_j = h^1 + u^1$ for some vertical $1$-chain $u^1$.
Then, we have $v^1 + h^1 = v^1 + u^1 + \partial \sum_j b^1_j \otimes s^1_j$.
Since $v^1 + h^1$ is closed, $v^1 + u^1$ is also closed.
Therefore, $v^1 + h^1$ is homologous to $v^1 + u^1$, a vertical $1$-cycle.

But now we can show that any vertical $1$-cycle of the form $a^0 \otimes s^1$ with $\partial s^1 = 0$ is a boundary:
The assumption~(iv) gives a base $1$-chain $c^1$ such that $a^0 = \partial c^1$.
Then,
$\partial ( c^1 \otimes s^1) = \sum_{t^0 \in \partial c^1} t^0 \otimes \varphi(c^1,t^0) s^1$.
Since $t^0 \otimes \varphi(c^1,t^0) s^1$ is homologous to $t^0 \otimes s^1$ by (v),
we see $\partial(c^1 \otimes s^1)$ is homologous to $(\partial c^1) \otimes s^1 = a^0 \otimes s^1$.
\end{proof}
This completes the proof of \Cref{lem:H1iso} for homology.
The cohomology isomorphism is straightforward
by abstract nonsense using the commutative diagram~\eqref{eq:comm-diagram}.
The details are as follows.

If $w_1$ is a $1$-cocycle of the base, then
$(\partial^\transp \Pi^* w_1)(e^2) = w_1(\Pi \partial e^2)
= w_1( \partial \Pi e^2 ) = (\partial^\transp w_1)(\Pi e^2) = 0$
for any bundle 2-chain $e^2$.
If $w_1 = \partial^\transp w_0$ is a chain of the base,
then
$(\Pi^* w_1)(e^1) = w_1( \Pi e^1) = w_0( \partial \Pi e^1 )
= w_0( \Pi \partial e^1 ) = (\partial^\transp\Pi^* w_0)(e^1)$,
and so $\Pi^* w_1$ is a coboundary.
This shows that $\Pi^*$ is a well-defined map from $H^1(\calB)$ to $H^1(\calE)$.
We claim that $\Pi^*$ is one-to-one.
This will finish the proof of \Cref{lem:H1iso} because the vector space dimensions
are the same for homology and cohomology.
To show the claim,
suppose that $\Pi^* w_1$ is a coboundary for a base $1$-cohomology representative $w_1$.
That is, $\Pi^* w_1$ vanishes on any bundle $1$-cycle.
Now by the homology isomorphism,
any bundle $1$-homology cycle is a lift of a base $1$-cycle.
This means that $w_1$ vanishes on all base $1$-cycles.
Since the dual vector space of $H_1(\calB)$ can be identified with $H^1(\calB)$,
we see $w_1$ is a coboundary.
\end{proof}

Under the assumptions of \Cref{lem:H1iso},
it will be instructive to have a more elementary description of the homology and cohomology representatives.
A lift of a base homology representative $c^1$ is illustrated in the proof above.
To recap, a lift consists of horizontal $1$-cells that project onto~$c^1$
together with vertical $1$-cells that cap the boundary of these horizontal $1$-cells
through a $1$-chain in the fiber.
Because of the twists, the boundary of the horizontal $1$-cells can be ``far apart'' within each fiber.

Regarding cohomology, since $\calB$ is a $1$-complex, there is no restriction on the $1$-cocycles;
every $1$-chain is a $1$-cocycle.
That is, any $1$-cell $b$ of the base represents a $1$-cohomology class.
Its lift functional $\Pi^* b$ as a cohomology representative
must evaluate to $1 \in \FF$ for every horizontal cell over $b$,
so the functional $\Pi^* b$ can be identified with $\sum_{f^0} b \otimes f^0$
where the sum is over all $0$-cells $f^0$ of the fiber.
The cohomology representative gives an upper bound on $d_X$
when the quantum code's $X$ logical operators are $1$-cohomology.
We summarize this as a proposition for later reference.
\begin{proposition}\label{prop:cohorep}
Assume all  conditions (i)---(v) of \Cref{lem:H1iso}.
If $\{[b_1], [b_2],\ldots\} $ is a basis of~$H^1(\calB)$ where $b_j$ are some base $1$-cocycles,
the following is a complete basis of representatives for~$H^1(\calE)$:
\[
\{ b_1 \otimes F_0, b_2 \otimes F_0, \ldots \},
\]
where $F_0 \in \calF_0$ denotes the sum of all fiber $0$-cells.
In particular, there exists a nontrivial representative of $H^1(\calE)$ of weight
equal to the number of $0$-cells in the fiber.
\end{proposition}

\subsection{Circle bundle over classical codes}\label{sec:circle-bundle}
\label{sketchsection}
We choose the fiber to be a cycle graph, i.e., a circle.
As a chain complex, the fiber is a $1$-complex
with $\mF$ $0$-cells and $\nF$ $1$-cells where $\mF = \nF > 1$.
This circle admits an automorphism group that is the dihedral group of order $2\nF$;
however, we do not use the reflection symmetry,
but only the rotation symmetry.
The circle fulfills all the conditions related to the fiber, namely (ii), (iii), and (v) of \Cref{lem:H1iso};
any automorphism of the circle leaves the fundamental homology cycle invariant.

\begin{definition}
Our fiber bundle code is a quantum CSS code
whose logical operators are associated with homology and cohomology
at dimension~$1$ of the twisted bundle complex $\calE_2 \to \calE_1 \to \calE_0$
built from the circle fiber $\calF_1 \to \calF_0$ and a base $\calB_1 \to \calB_0$.
\end{definition}
\noindent
This definition leaves room for the base complex to be any classical code
and for the twists to be completely arbitrary members of $\ZZ_{\nF}$.

We will pick $\nF=\mF=\ell^2$ for some integer $\ell$ (and it will be convenient, though not strictly necessary, to assume $\ell$ is odd, hence $\nF$ is odd).  In fact, all twists will be multiples of $\ell$ so we only use a subgroup of order $\nF/\ell=\ell$ of the rotation group.  Roughly speaking, this is done so that if some weight needs to ``move through the fiber" to join two cells that differ by a twist, we have some lower bound on how far it needs to move.

We will use a random classical code for the base~$\calB$.
There will be $\nB$ bits (variables) in this classical code, considered to be $1$-cells of the base, and there will be $\mB$ parity checks in the code, considered to be $0$-cells in the base.
We will represent this classical code by its Tanner graph, a bipartite graph~$B$  with $\mB$ left-vertices and $\nB$ right-vertices.
We will later choose $\mB = (3/4) \nB$, and all vertices will have degree very close to~$\Delta = \Theta(\log^2 \nB)$.
In this way the random classical code will have minimum distance~$\Omega(\nB)$ and all of its parity checks will be linearly independent (with high probability); the latter condition is equivalent to $H_0(\calB) = 0$.

Thus the bundle $\calE$ will have $N = \nB\cdot \mF+\mB\cdot\nF$ $1$-cells corresponding to qubits of the resulting quantum code, and the total number of cells in the bundle will be $(\nB + \mB) \cdot (\nF + \mF)$.
Given that $H_0(\calB) = 0$, i.e., $b_0(\calB) = 0$, it then follows by construction that $b_1(\cB)=(1/4)\nB$ so that the fiber bundle code has $\Theta(\nB)$ logical qubits.
We will prove that the resulting quantum code has distances $d_X=\Omega(\mF/\log^2 \nB)$ (see \Cref{cohomweight}) and $d_Z=\Omega(\nB \cdot \mF^{1/2} / \log^2 \nB)$ (see \Cref{homweight}) provided $\nB \geq \mF$.
We then choose $\nB \sim \mF$, giving $N=\Theta(\nB^2)$ and giving distances $d_X=\Omega(N^{1/2}/\log^2 N)$, $d_Z=\Omega(N^{3/4}/\log^2 N)$.                                                                                                  Together, these facts prove \Cref{mainth}.

\paragraph{Remark.} It is possible to slightly improve the polylogs in \Cref{mainth} by  adjusting our choices for $\ell,\nF,\nB$ by polylogarithmic factors.
Specifically, a minor improvement arises from choosing $\mF = \nB/\Delta$ and $\ell = \sqrt{\mF/\Delta}$, where recall $\Delta = \Theta(\log^2 \nB)$.
However, these polylog improvements deteriorate again after passing through the weight reduction process leading to \Cref{coro}.
Thus we have chosen to make slightly non-optimal parameter choices so as to simplify the presentation.

\section{The random base code, with twists}
\label{randbase}
\newcommand{\sch}{S}
\newcommand{\kreg}{k}
\newcommand{\el}{\ell}
\newcommand{\heads}[1]{\scalebox{.75}[1.0]{\textsc{Heads}}_{#1}}
\newcommand{\tails}[1]{\scalebox{.75}[1.0]{\textsc{Tails}}_{#1}}
\newcommand{\pB}{B}

Throughout this section and the next section we write $n = \nB$ and $m = \mB$, for brevity.

\subsection{A random base code} \label{sec:random-base}
    The base code $\pB$ is identified with its Tanner graph, having variable vertices~$[n]$ and check vertices~$[m]$.
    Our construction will require $\tfrac12 < m/n < 1$; for simplicity we fix
    \begin{equation}    \label{eqn:m34}
        m = \tfrac34 n,
    \end{equation}
    assuming that $m$ is an integer.
    We will also fix a parameter $\Delta = \Delta(n)$ representing the average degree of the check vertices.
    Eventually we will choose
    \[
        \Delta = \Theta(\log^2 n),
    \]
    but for now we only assume
    \begin{equation}    \label{eqn:Delta}
        \beta \ln n \leq \Delta \leq n^{o(1)},
    \end{equation}
    where $\beta$ is a large universal constant to be chosen later.
    We will choose a \emph{random} base code~$\bB$,\footnote{In this section, boldface denotes random variables/objects.} with the neighborhood $\bdry^\transp a$ of each check $a \in [m]$ independently being a random density-$\frac{\Delta}{n}$ subset of~$[n]$.
    By this we mean that each variable $i \in [n]$ is included into~$\bdry^\transp a$ independently with probability~$\frac{\Delta}{n}$.

It is well known that such a randomly constructed code~$\bB$ will have various expansion-type properties.
We collect here some standard results along these lines.

\begin{proposition}                                     \label{prop:bounded-degree}
    Except with probability at most $O(1/n^{100})$, all check vertices in~$\bB$ have degree between $.99\Delta$ and~$1.01\Delta$ and all variable vertices have degree between $.74\Delta$ and~$.76\Delta$.
\end{proposition}
\begin{proof}
    This can be achieved by a standard Chernoff + union bound argument, taking the constant~$\beta$ in \cref{eqn:Delta} large enough and using $m \leq n$.
\end{proof}

\begin{proposition}                                     \label{prop:cohom-unique}
    The bipartite graph $\bB$ has the following property, except with probability at most $O(1/n^{100})$:
    For every $S \subseteq [m]$ with $|S| \leq \frac{1}{\XConst\Delta} m$, the neighborhood of~$S$ in~$[n]$ has cardinality at least $.9 \Delta |S|$.
\end{proposition}
\begin{proof}
    By \Cref{prop:bounded-degree}, it suffices to prove this conditioned on the assumption that every check vertex in $\bB$ has degree between $.99\Delta$ and~$1.01\Delta$.
    Under this conditioning, the neighborhoods of each check vertex remain independent and have a certain distribution on their cardinality; conditioned on their cardinality, they are uniformly random subsets of~$[n]$.
    By ignoring edges (which only hurts us), we may therefore assume that the neighborhoods of the check vertices are independent random subsets of~$[n]$ of cardinality $d \coloneqq \lceil .99 \Delta \rceil$.
    Thus we have reduced to the $d$-regular model of random bipartite graphs, where the claim we want to prove is standard, relying on the inequality
    \[
        .99\Delta > \frac{h_2(\tfrac{1}{\XConst\Delta}) + h_2(\tfrac{.8 \Delta}{\XConst\Delta})}{h_2(\tfrac{1}{\XConst\Delta}) - \tfrac{.8}{\XConst} h_2(\tfrac{1}{.8\Delta})},
    \]
    (here $h_2(\cdot)$ is the binary entropy function), and on the assumption from \cref{eqn:Delta} that $\Delta \geq \beta \ln n$ for large constant~$\beta$; see, e.g.,~\cite{Chu79,Bas81}.
\end{proof}
\begin{corollary}                                       \label{cor:counique}
    Assume $B$ satisfies the conclusion of \Cref{prop:bounded-degree,prop:cohom-unique}.
    Let $S \subseteq [m]$ have $0 < |S| \leq \frac{1}{\XConst\Delta} m$.
    Say a variable vertex $j \in [n]$ is a \emph{counique neighbor of~$S$} if it neighbors exactly one vertex in~$S$.
    Then, the number of non-counique neighbors of $S$ is at most $.09 \Delta |S|$ and
    there exists some  $a^\bullet \in S$ for which more than $.81\Delta$ of its neighbors are counique neighbors.
    In particular, $a^\bullet$ has at least a~$.8$ fraction of its neighbors (a strict majority) being counique.
\end{corollary}
\begin{proof}
    Let every $a \in S$ give a token to each of its neighbors in~$[n]$.
    By \Cref{prop:bounded-degree}, at least $.99\Delta|S|$ tokens are given out.
    Every vertex in the set $C$ of counique neighbors gets one token, and every vertex in the set $C'$ of non-unique neighbors gets at least two.
    Thus $.99\Delta |S| \geq |C| + 2|C'|$.
    But $|C| + |C'| \geq .9 \Delta |S|$, by \Cref{prop:cohom-unique}.
    Thus $|C'| \leq .09\Delta|S|$, and so $|C| \geq .81 \Delta |S|$.
    We conclude that even on average, a vertex $a \in S$ gives out at least $.81 \Delta$ tokens to counique neighbors.
\end{proof}

\begin{proposition}                                     \label{prop:lin-indep}
    Except with probability $O(1/n^{100})$, the parity check matrix of~$\bB$ is of full rank.
\end{proposition}
\begin{proof}
    Let $\bv_a \in \F_2^n$ be the indicator for the neighborhood of check vertex~$a \in [m]$.
    We wish to show for all $\emptyset \neq A \subseteq [m]$ that $\sum_{a \in A} \bv_a \neq 0$.
    So fix such an~$A$ of cardinality~$w \neq 0$.
    For $i \in [n]$, the $i$th coordinate of $\sum_{a \in A} \bv_a$ is distributed as $\text{Binomial}(w, \frac{\Delta}{n})$ modulo~$2$.
    The event that this is~$0$ has probability
    \[
        q_w \coloneqq \tfrac12 + \tfrac12(1 - \tfrac{2\Delta}{n})^w
    \]
    and these events are independent across $i \in [n]$.
    For $w \leq \frac{n}{\Delta}$ it holds that $q_w \leq \exp(-\frac{w\Delta}{2n})$, and hence~$q_w^n$ (which is the probability of $\sum_{a \in A} \bv_a = 0$) is at most $\exp(-\frac{\Delta}{2})^w \leq 1/n^{101w}$, the last inequality provided the constant~$\beta$ in \cref{eqn:Delta} is large enough.
    On the other hand, if $w \geq \frac{n}{\Delta}$ then we have $q_w \leq \tfrac12 + \tfrac12\exp(-\tfrac{2\Delta w}{n}) \leq \tfrac12 + \tfrac12 e^{-2} \leq 2^{-.8}$, and hence $q_w^n \leq 2^{-.8n}$.
    Taking a union bound over all~$A$, we conclude that the probability of $\bB$'s parity check matrix having a nontrivial linear dependence is at most
    \[
        \sum_{1 \leq w \leq \frac{n}{\Delta}} \tbinom{m}{w}/n^{101w} +
        \sum_{\frac{n}{\Delta} \leq w \leq m} 2^{-.8n} \leq \parens*{(1+1/n^{101})^{m}-1} + 2^m 2^{-.8n} = O(1/n^{100}),
    \]
    where the last inequality used $m = \frac34 n$.
\end{proof}

\begin{proposition}                                     \label{prop:min-dist}
    Except with probability $O(1/n^{100})$, the code $\bB$ has minimum distance at least~$.2n$.
\end{proposition}
\begin{proof}
    Indeed, the number~$.2$ can be replaced with any $\delta$ such that $h_2(\delta) \leq \frac34$; in other words, with high probability~$\bB$ achieves the Gilbert--Varshamov bound.
    This is a standard property of random LDPC codes with~$\Delta \to \infty$.
    Gallager showed it in a slightly different model of~$\Delta$-regular random LDPC codes; the proof in our case is a standard exercise along the lines of \Cref{prop:lin-indep} and \Cref{lem:key-light}.
\end{proof}

\subsection{Cohomology representative weight}
\label{cohomd}

We have not yet specified the twists,
but we are already able to lower-bound the weight of a cohomology representative.
Recall \Cref{prop:cohorep} implies that the least weight cohomology represenative
has its weight \emph{upper}-bounded by~$\mF$.
We start with a definition, and then give the lower bound.
\begin{definition}  \label{def:shadow}
    In brief, the \emph{shadow} of a chain $e \in \calE$ is the set of base cells on which~$e$ has support.
    More precisely, given a bundle $1$-chain $e$, write it as a sum $h+v$ of horizontal and vertical chains, 
    with $h = \sum_b b \otimes f^0_b$ and $v = \sum_a a \otimes f^1_a$ where $a$ runs over base $0$-cells and $b$ over base $1$-cells.
    Then the shadow of the vertical part of $e$ is the set $\{a  : f^1_a \neq 0\}$,
    and the shadow of the horizontal part of $e$ is the set $\{ b : f^0_b \neq 0 \}$.
    Similarly given a bundle $0$-chain $e$ with $e=\sum_a a \otimes f^0_a$, its shadow is $\{a : f^0_a\neq 0\}$.
    We write $|e|_{\mathrm{vsw}}$ for the \emph{vertical shadow weight} of a bundle $1$-chain $e$; 
    i.e., the cardinality of its shadow of the vertical part.
    For a $0$-chain $e$, we simply write $|e|_\mathrm{sw}$ for the shadow weight.
\end{definition}

\begin{lemma}
\label{cohomweight}
Assume the base code~$B$ satisfies the conclusions of \Cref{prop:bounded-degree,prop:cohom-unique}.
Then for \emph{any} choice of twists, the following holds:
Let $r$ be any nontrivial representative of $H^1(\cE)$ and write it as a sum $h+v$ of horizontal and vertical chains.
Then either $|h| \geq \mF/2$, or $|v|_{\mathrm{vsw}} \geq \Omega(n/\Delta)$.
In particular, $|r| = |h| + |v| \geq \Omega(\mF + n/\Delta)$, which is $\Omega(\mF/\Delta)$ assuming $n \geq \mF$.
\end{lemma}
\begin{proof}
Let $r$ be any nontrivial representative of $H^1(\cE)$.
Using the basis of \Cref{prop:cohorep}, we may write $r = x \otimes F_0 + \bdry^\transp u$ where $x$ is a nontrivial representative of $H^1(\cB)$ and where $u$ is a bundle $0$-chain.
This expression is not unique; if we toggle \mbox{$u \mapsto u + a \otimes F_0$} for any base $0$-cell $a$ and simultaneously change $x \mapsto x + \bdry^\transp a$, then the overall change in $x \otimes F_0 + \bdry^\transp u$ is $(\bdry^\transp a)\otimes F_0 + \bdry^\transp (a \otimes F_0)$, which vanishes regardless of twists because $F_0$ is invariant under any fiber automorphism.
Thus we may assume that in the expansion $u = \sum_a a \otimes f_a$ (where each $a$ is a base $0$-cell and each $f_a$ is a fiber $0$-chain), each $f_a$ has Hamming weight at most~$\mF/2$.
Having done this, the shadow of the vertical part of $r$ and the shadow of~$u$ coincide (as no $f_a$ equals~$F_0$).
Writing $S$ for this shadow, we have
\begin{equation}
r
=
x \otimes F_0 + \bdry^\transp u
= \underbrace{x \otimes F_0 + \sum_{a \in S, b \in \bdry^\transp a} b \otimes \varphi(b,a)^{-1} f_a}_{h} +
\underbrace{\sum_{a \in S} a \otimes \bdry^\transp f_a}_{v} \label{eq:cohor}
\end{equation}
where $h \in \calE_{1,0}$ is horizontal and $v \in \calE_{0,1}$ is vertical.
Here $\varphi(b,a)$ are some twists.
If $|S| \geq \frac 1 {\XConst\Delta} m$ then we are done easily: for every $a \in S$ we have $0 \neq |f_a| \neq \mF$ and hence $|\bdry^\transp f_a| \ge 2$; thus $|r| \geq |v| \geq 2|S| \geq \Omega(\mF/\Delta)$, as needed.

Thus it remains to handle the case that $|S| \leq \frac 1 {\XConst\Delta} m$.
In this case we claim that in fact $|h| \geq \mF/2$, which is more than sufficient to complete the proof.
First, if $S = \emptyset$ then $|r| = |h| = |x| \mF \ge \mF$, using the fact that $r$~is a \emph{nontrivial} cohomology representative, and the claim is established.

Otherwise, $0 < |S| \leq \frac{1}{\XConst\Delta}m$, and since $B$~satisfies the conclusion of \Cref{prop:cohom-unique}, the number of neighbors (base $1$-cells) of $S$ is at least $.8 \Delta |S|$.
As every $a \in S$ has between $.99\Delta$ and $1.01\Delta$ neighbors (\Cref{prop:bounded-degree}), it follows that there must be some base $0$-cell $a^\bullet \in S$ such that more than half of the $1$-cells in the coboundary of $a^\bullet$ are not in the coboundary of any other $0$-cell of~$S$.
(This relies on $.8 > 1.01\cdot\frac34$ and also $|S| \neq 0$.)
Let $C$ be this set of \emph{counique}-neighbor base $1$-cells:
$C = \bdry^\transp a^\bullet \setminus \bigcup_{a \in S \setminus \{a^\bullet\}} \bdry^\transp a$.
We now consider whether or not $C$ overlaps with~$x$.

If~$C \cap x$ contains a base 1-cell $b$,
then the number of all horizontal cells of $r$ over $b$ is already at least~$\mF/2$,
because $b \otimes \varphi(b,a)^{-1}f_a$ has at most~$\mF/2$ cells.

Otherwise, if $C \cap x = \emptyset$, then dropping $a^\bullet$ from $S$ would reduce $|h|$,
since less than half of the horizontal part of $\bdry^\transp a^\bullet$ was canceling in \cref{eq:cohor}
by the choice of $a^\bullet$.
This dropping yields a cohomologous representative $r'$ of $H^1(\calE)$
with a lighter horizontal part and where the shadow~$S$ of the vertical part still satisfies $|S| \leq \frac 1 {\XConst\Delta} m$.
Repeating this argument, we either come to a representative whose horizontal weight is at least~$\mF/2$,
or else we reduce to the case of~$S = \emptyset$ where it was already shown that the horizontal part has weight at least~$\mF$.
Either way, we have established the claim that the original~$r$ had $|h| \geq \mF/2$.
\end{proof}

\subsection{Twists}

We now choose the twists.
It is worth making a mental shift at this point.
The twist is defined by an automorphism $\varphi(b^1,v)$
which is a function of a $1$-cell $b^1$ and some $0$-cell $v\in \bdry b^1$.
Mentally, up to this point, we have tended to think of it as ``for each $1$-cell $b^1$, for each $v$ in $\bdry b^1$" there is a twist, but now it is worth thinking instead ``for each $0$-cell $v$, for each $b^1$ in $\bdry^\transp v$" there is a twist.  This of course is no difference mathematically but is an easier mental picture.
In the language of the base code, it means that for each check of the base code, for each bit in the check, there is a twist.

All twists will be chosen to be integer multiples of $\ell$.
Indeed, we will only have $\kreg = \Theta(\log n)$ distinct choices of twists (where the precise~$\kreg$ will be specified later).
Informally, we will partition the checks of the base code into $\kreg$ ``types''; for each type, we will pick a single twist by some random multiple of $\ell$, and for each bit in each check, we will either twist by that multiple or by~$0$, the choice of which again being random.

Formally,
let us make the following definition:
\begin{definition}
    Given a base code~$B$, we say it is \emph{partitioned} if:
    \begin{itemize}
        \item its check vertices are partitioned into $\kreg$ sets $T_1, \dots, T_\kreg$ of equal size~$m/\kreg$ (assumed to be an integer), one for each ``type'';
        \item and, the neighborhood $\bdry^\transp a$ of each check $a \in [m]$ is partitioned into two sets, $\heads{a}$ and $\tails{a}$.
    \end{itemize}
\end{definition}

In our construction, in addition to assuming that the base code $\bB$ has random density-$\frac{\Delta}{n}$ checks as in \Cref{sec:random-base}, we also assume:
\begin{itemize}
    \item the partition into types is the trivial one, $T_1 = \{1, \dots, m/\kreg\}$, $T_2 = \{m/\kreg + 1, \dots, 2m/\kreg\}$, etc.\ (in fact this doesn't matter since $\bB$ is random in the first place);
     \item and, the neighborhood  $\bdry^\transp a$ of each check $a \in [m]$ is partitioned uniformly at random into $\heads{a}$ and $\tails{a}$.
\end{itemize}

Finally, for each set $T_\tau$ we choose the twists $\varphi_\tau$ so that a certain graph defined later is a good spectral expander.
It will be shown that choosing the
$\varphi_{\tau}$ uniformly at random from \mbox{$\ell,2\ell,\ldots,(\ell-1)\ell$} 
will give the desired spectral expansion with high probability (which can then certified if desired).

Then, given this partitioning we define for any pair $b^1,a$ with $a\in \partial b^1$,
the twist $\varphi(b^1,a)$ to be given as follows:
check node $a$ is in some set $T_\tau$.
If $b^1$ is in $\heads{a}$ then $\varphi(b^1,a)=0$.
Else if $b^1$ is in $\tails{a}$ then $\varphi(b^1,a)=\varphi_\tau$.

\subsection{Homology representative weight}
\label{homd}
We now prove bounds on the least weight representative of homology.
We will reduce this problem to proving certain properties of a classical code
which we study in \Cref{codingboundsection}.

We use the assumption that the twists are multiples of $\ell$.
Given a horizontal $1$-cell $b \otimes f$,
we say that $f \in \{0,1,\ldots,m-1\}$ is the \emph{fiber position} of $b \otimes f$.
Let us first simplify homology representatives:
\begin{lemma}\label{lem:sliding}
For any $1$-homology representative of weight $w$,
there exists a homologous representative $r$ of weight at most $w$
such that the fiber position of any nonzero horizontal cell in the support of~$r$ is \mbox{$0 \bmod \ell$}.
\end{lemma}
\begin{proof}
To prove this, we use a method we call ``sliding.''
Consider a graph%
\footnote{This graph $\sanH$ has nothing to do with any other graph in this paper.}
$\sanH$
whose nodes correspond to the horizontal cells of a given homology representative $s$
and with a link between two horizontal cells if and only if their boundaries overlap.
The fiber positions of a pair of linked horizontal cells differ by $0 \bmod \ell$
because the twists are multiples of $\ell$.

If $\sanH$ is connected, then we can translate $s$ along the fibers,
i.e., all fiber cells $f^0_j$ and $f^1_i$ of $s = \sum_j b_j \otimes f^0_j + \sum_i a_i \otimes f^1_i$
are replaced by shifted fiber cells $f^0_j + y$ and $f^1_i +y$ with a common $y = \pm 1 \in \ZZ_\nF = \ZZ_\mF$.
If we keep translating the representative
until the fiber position of any horizontal cell is $0 \bmod \ell$,
then we have the claim of the lemma.

If $\sanH$ has more than one connected cluster (a maximal connected subset of nodes),
then we ``slide'' any one cluster $C$ of horizontal cells in two steps as follows.
First, we translate all the horizontal cells $h$ in $C$
by shifting their fiber components by $\pm 1$ as above.
Second, we add vertical cells on the fibers over base $0$-cells
where the boundary $\partial \sum_{h \in C} h$ is supported,
so that the overall chain is still closed.
We need one and only one vertical cell around each $0$-cell of $\partial \sum_{h \in C} h$.
This sliding is clearly a modification of the original representative $s$
by the boundary of a $2$-chain (a $Z$-stabilizer).
Every added vertical cell in the second step contributes to $\pm 1$
to the weight of vertical part of the homology representative.
In fact, the total weight is a piecewise linear function of the amount that the cluster is slid,
with the slope of this function constant until the cluster becomes connected to another cluster,
at which point we modify the graph $\sanH$ by adding links.

Hence, we end up with a single cluster that contains all the horizontal cells of
the slid homology representative, and we finish by an overall translation as before.
\end{proof}

An intuitive picture for this sliding is as follows: there are $1$-chains in the fiber whose endpoints are such as to cancel the boundary of $h$; we can think of these $1$-chains as ``strings" that are ``pulling" on $h$; we slide in the direction in which the strings pull most strongly (or pick a direction arbitrarily if there is no preferred direction).
Once the cluster becomes connected to another cluster, we slide that combined cluster, and so on.
Continue this until there is only a single connected cluster.  All the vertices of the cluster must the same fiber position mod $\ell$; finally, slide that cluster until the claim is obeyed.

The sliding simplifies the problem of the weight of homology cycles as follows.
Note that the base code's minimum distance is $\Theta(n)$ with high probability; see \Cref{prop:min-dist}.
\begin{lemma}
\label{wasslid}
Assume that the base code has distance $\Theta(n)$.
The weight of a nontrivial homology representative is lower bounded by the minimum of the following problem.

Consider a chain
$h=\sum_{b} b \otimes f_b$
consisting of horizontal cells whose fiber positions are $0 \bmod \ell$.
For any $0$-cell $a$ in the base, say that $a$ \emph{has an error}
if $\partial h$ is nonvanishing in the fiber over $a$.
Then, minimize $|h|+\ell |\partial h|_\mathrm{sw}$,
subject only to the requirement that there are $\Theta(n)$ different $b$ such that $f_b \neq 0$.
\begin{proof}
Write a nontrivial $1$-homology cycle $r$ as $r = h + v$ where $h$ is the horizontal part and $v$ the vertical part.
Let $h = \sum_b b \otimes f_b$, where the sum is over base $1$-cells $b$ and $f_b$ is a fiber $0$-chain.
The bundle projection $\Pi(h)$ is a nontrivial element of $H_1(\cB)$ by \Cref{lem:H1iso}.
Since the base code has distance $\Theta(n)$,
then the set of $b$ such that $f_b \neq 0$ has cardinality $\Theta(n)$.

Slide as above.
Then, the weight $|v|$ of the vertical part of $r = h + v$ is at least $\ell$ times the number of base 0-cells that have an error;
here we use that after sliding all the nonzero ``strings" in $v$ must have length at least $\ell$.
\end{proof}
\end{lemma}

This lemma implies a significant simplification of the problem:
it suffices to consider just horizontal chains.
\Cref{thm:ltc} proven below implies that any horizontal $1$-chain $h$ of
a sufficiently small weight $|h|$
compared to $n\ell/\Delta$ will have at least $\Omega(|h|)$ errors.
Hence for any nontrivial homology representative $r = h + v$
with the horizontal part $h$ and the vertical part $v$
we have either $|h| \geq n\ell/\Delta$
or $|v| \geq \Omega(n \ell)$ where the latter case
is because we must have $|h|=\Omega(n)$.
Hence:
\begin{lemma}
\label{homweight}
Assume $m=\Theta(n)$.
With high probability,
the weight of any nontrivial representative of $H_1(\cE)$ is $\Omega(n\ell/\Delta)$.
\end{lemma}

\Cref{thm:ltc} is however phrased in terms of a classical code that we call a {\it twist graph code}.
This is the error correcting code whose bits are horizontal cells of $\cE$ and whose checks are obtained from $0$-cells of $\cE$ in the obvious way.
However, we will find it useful to explicitly redefine this code in terms of a graph that we call the twist graph in order to use expansion properties.

\section{Coding Bounds for Twist Graph Code}
\label{codingboundsection}
We now define the twist graph code.
This is simply a restatement of the checks on horizontal $1$-chains in more graph theoretic terms.
We are concerned with checks on horizontal $1$-chains which obey the condition of \Cref{wasslid} that
$h=\sum_{j} j \otimes f_j$ with
all $(f_j)_i=0$ unless $i=0\bmod \ell$.
So, throughout this section, we will regard these $1$-chains as bit strings of length $n\ell$ rather than of length $n \cdot n_F$.

\subsection{Twist graph and assumptions on twists}
We define the twist graph $\vec{\sch}$ as follows.  This is a directed graph, possibly with multi-edges (but without self-loops).
We have vertices chosen from $[\ell]$.
For each $t \in [\kreg]$, there is an edge from each vertex $i$ to \mbox{$i+\varphi_t/\ell \pmod \ell$}.
Note: the twist $\varphi_t$ is a multiple of $\ell$ so $\varphi_t/\ell \in \{1,2,\ldots,\ell-1\}$.

While the twist graph is obtained as described in the above paragraph,
we will in this section only use the following more general assumptions on $\vec{\sch}$:
\begin{definition}
    It is assumed we have a directed graph $\vec{\sch}$, with vertex set $U$ of cardinality~$\el$, and $\kreg$ ``types'' of directed edges.
    This graph may have multi-edges but no self-loops.
    We assume each vertex in $\vec{\sch}$ has one in-edge of each type, and also one out-edge of each type.
\end{definition}

\begin{notation}
    We write $\sch$ for the undirected version of~$\vec{\sch}$, which is a $2\kreg$-regular graph.
    Letting \mbox{$1 = \kappa_1 \geq \kappa_2 \geq \cdots \geq \kappa_\el \geq -1$} denote the eigenvalues of the normalized adjacency matrix of~$\sch$, we write $\kappa_{\sch} = \max\{\kappa_2, |\kappa_\el|\}$ for the second-largest in magnitude.
\end{notation}

We will need expansion properties of this graph $\sch$.
From~\cite{AR94}, for any $\kappa > 0$, for $\kreg=O(\log(\ell)/\kappa^2)$, with probability at least~$.999$ we have $\kappa_{\sch} \leq \eps$.
We will ultimately choose~$\eps$ to be a very small universal constant (see \Cref{asmp:kapp}); thus $\kreg = O(\log \ell) = O(\log \mF)$.

\subsection{Twist graph code}
\begin{definition}
    Given $\vec{\sch}$ and a partitioned base code~$\pB$, we define a new $\F_2$-linear code $\pB(\vec{\sch})$ as follows:
    The block length of $\pB(\vec{\sch})$ is $\el \cdot n$, and we write a received word $w \in \F_2^{U \times [n]}$ as $w = (w_u)_{u \in U}$, where $w_u \in \F_2^n$.
    The number of parity checks in $\pB(\vec{\sch})$ will be $\el \cdot m$, and they are defined as follows.
    Given a type $\tau \in [\kreg]$ and a pair $(y,z) \in \F_2^n \times \F_2^n$, we introduce the notion of the \emph{type-$\tau$ checks on $(y,z)$}, meaning all $m/\kreg$ parity checks of the form
    \begin{equation}    \label{eqn:ht-check}
        \sum_{i \in \tails{a}} y_i + \sum_{j \in \heads{a}} z_j = 0 \pmod{2} \quad \text{for $a \in T_\tau$.}
    \end{equation}
    Now $\pB(\vec{\sch})$ consist of imposing, for each directed edge $(u,v)$ in $\vec{\sch}$ of type~$\tau$,  all type-$\tau$ checks on $(w_u, w_v)$.
\end{definition}

Note that for a fixed~$a$, each of $\heads{a}$ and $\tails{a}$ is a random density-$\frac{\Delta}{2n}$ subset of~$[n]$, but these two sets are not quite independent (since any $j \in [n]$ is in at most one of them).
We will use the following fact:
\begin{proposition}                                       \label{prop:chi}
    Let $(y,z) \in \F_2^n \times \F_2^n$.
    For a randomly chosen and randomly partitioned check imposed on~$(y,z)$, the probability it is violated is
    \[
	    \frac12 - \frac12 \left(1-\frac{\Delta}{n}\right)^{|y +z|}\left(1-\frac{2\Delta}{n}\right)^{|y \cap z|}
        \geq \frac12 - \frac12 \exp\parens*{-(|y|+|z|)\frac{\Delta}{n}}.
	\]
\end{proposition}
\begin{proof}
    Call the check~$a$.
    We use the notation $\chi_J(x) = \prod_{j \in J} (-1)^{x_j}$ and write $\bI = \tails{a}$, $\bI' = \heads{a}$ for brevity.
    From \cref{eqn:ht-check}, the probability of violation is
    \[
        \E\bracks*{\tfrac12 - \tfrac12\chi_{\bI}(y)\chi_{\bI'}(z)}
        = \tfrac12 - \tfrac12 \prod_{j=1}^n \E[\chi_{\bI_j}(y_j) \chi_{\bI'_j}(z_j)],
    \]
    where we used independence of the pairs $(\bI_j, \bI'_j)$ across $j \in [n]$.
    By first considering whether or not $\bdry^\transp a \ni j$, and if so, whether $j \in \bI$ or $j \in \bI'$, we compute
    \[
        \E[\chi_{\bI_j}(y_j) \chi_{\bI'_j}(z_j)] = \begin{cases}
                                                                            1 & \text{if $y_j = z_j = 0$,}\\
                                                                            1 - \tfrac{\Delta}{n} & \text{if $y_j \mathmakebox[\widthof{$\displaystyle {}={}$}]{{}+{}}
                                                                            z_j = 1$,} \\
                                                                            1 - \tfrac{2\Delta}{n} & \text{if $y_j = z_j = 1$.}
                                                                      \end{cases}
    \]
    Thus our expression for the exact probability of violation follows.
    As for the inequality, it uses
    \[
        \left(1-\frac{\Delta}{n}\right)^{|y +z|}\left(1-\frac{2\Delta}{n}\right)^{|y \cap z|}
        \leq \exp(-|y+z|\Delta/n) \exp(-2|y \cap z| \Delta/n)
    \]
    and the fact that $|y|+|z| = |y+z| + 2|y \cap z|$.
\end{proof}

Below we establish a property of $\bB(\vec{\sch})$ that is minor variant of a standard property of random LDPC codes, that words of small Hamming weight violate a proportionate number of parity checks.
We first need to upgrade the assumption \cref{eqn:Delta}:
\begin{assumption}  \label{asmp:del}
    We now make the stronger assumption
    \[
        \beta \kreg \ln n \leq \Delta = \Delta(n) \leq n^{o(1)},
    \]
     where $\beta$ is a large universal constant to be fixed later.
     As a remark, since $\kreg$ will eventually be $\Theta(\log n)$, the above constraint is the reason for our eventual choice of $\Delta = \Theta(\log^2 n)$.
\end{assumption}
\begin{lemma}                                     \label{lem:key-light}
    For random $\bB(\vec{\sch})$ as described, except with probability at most $O(1/n^{100})$  the following holds:
    For any type $\tau \in [\kreg]$ and any pair $(y, z) \in \F_2^n \times \F_2^n$ with $|y|,|z| \leq \frac{4}{\Delta} n$, the number of violated type-$\tau$ checks on $(y,z)$ is at least~$.01\frac{\Delta}{\kreg}(|y|+|z|)$.
\end{lemma}
\begin{proof}
    Fix any type~$\tau$ and any pair $(y,z)$ with $|y|+|z| = t$, where $0 < t \leq \frac{8}{\Delta} n$.
    Using \Cref{prop:chi}, 
    the probability that a given constraint $a \in T_\tau$ is violated is at least
    \[
        \tfrac12 - \tfrac12 \exp(-\Delta t/n) \geq \tfrac12 - \tfrac12 (1-.1\Delta t/n) > .04 \Delta t/n,
    \]
    where the first inequality used $t \leq \frac{8}{\Delta} n$.
    Thus if we consider all $m/\kreg = (3/4)n/\kreg$ randomly chosen checks in~$T_\tau$, the expected number that are violated is at least $\mu \coloneqq .03 \Delta t/\kreg$.
    A Chernoff bound therefore implies that
    \[
        \Pr[\text{$(y,z)$ violates fewer than $\tfrac13 \mu = .01\tfrac{\Delta}{\kreg}(|y|+|z|)$ checks from $T_\tau$}] \leq \exp(-\tfrac29\mu).
    \]
    Taking a union bound over all $\tau \in [\kreg]$ and all pairs~$(y,z)$ with $|y|+|z| = t$ (of which there are at most $\binom{2n}{t} \leq (2en/t)^t$) yields a failure probability of at most
    \[
        \kreg (2en/t)^t \exp(-\tfrac29\mu) \leq n^{o(1)} \exp(\ln(2en/t) - (.06/9) \Delta /\kreg)^t \leq 1/n^{100t},
    \]
    provided the constant~$\beta$ in \Cref{asmp:del} is large enough.
    Now taking a union bound over all~$t \geq 1$ completes the proof.
\end{proof}

Now we come to a subtler property of the code $\bB(\vec{\sch})$: Even if a received word $w = (w_u)_{u \in U}$ has some $w_u$ that does \emph{not} have low Hamming weight (it may even have Hamming weight~$n/2$), still there will be $\Omega(|w_u|)$ violated checks among \emph{all} the checks involving~$w_u$.
It is crucial here is that in \cref{eqn:m34} we ensured $2m > n$; in this way, when the $2m$ checks involving~$w_u$ are chosen at random, there is a chance for the union bound to overcome the~$2^n$ possibilities for~$w_u$.
(More precisely, when the $m$ checks are chosen at random and then split into $\heads{}/\tails{}$ pairs,
the probability of at least one violation in the pair is close to~$.75$ as shown in \cref{ineq:73} below,
so the probability of no violation in all pairs is roughly $.25^m = 2^{-2m}$.)

\begin{lemma}                                     \label{thm:key-heavy}
    For random $\bB(\vec{\sch})$ as described, except with probability at most $2^{-\Omega(n)}$  the following holds:
    Let $x \in \F_2^n$ be any word with $|x| \geq \frac{4}{\Delta} n$.
    Let $\Upsilon \subseteq [\kreg]$ have cardinality at least $.98 \kreg$.
    Assume we have a sequence of pairs $(x,y_\tau)_{\tau \in \Upsilon}$ and $(z_\tau,x)_{\tau \in \Upsilon}$  in $\F_2^{n} \times \F_2^n$, where $|y_\tau|, |z_\tau| \leq \frac{2}{\Delta} n$ for all $\tau \in \Upsilon$.
    Then, aggregating the  type-$\tau$ checks on $(x,y_\tau)$ and $(z_\tau,x)$ for all $\tau \in \Upsilon$, there are at least~$.01n$ parity check violations (out of a total of $2 |\Upsilon| m/k$).
\end{lemma}
\begin{proof}
    There are at most $2^n \cdot 2^\kreg \cdot (2^{h_2(2/\Delta) n})^{2\kreg} \leq 2^{(1+o(1)) n}$ choices for $x$, $\Upsilon$, the $y_\tau$'s,  and the~$z_\tau$'s (where the inequality used \Cref{asmp:del}).
    Let us fix any such choices and consider the random partitioned code~$\bB$.
    For a given type $\tau \in \Upsilon$, we have $m/\kreg$ randomly chosen and randomly partitioned checks that get imposed on $(x,y_\tau)$ and also on $(z_\tau,x)$.
    Consider one such check $a \in T_\tau$.
    Using the notation $\chi_J$ and $\bI, \bI'$ from \Cref{prop:chi}, the probability that the two checks~$a$ imposes on $(x,y_\tau)$ and $(z_\tau,x)$ are \emph{both satisfied} is
    \begin{multline*}
        \E[(\tfrac12 + \tfrac12\chi_{\bI}(x) \chi_{\bI'}(y_\tau))
                                   (\tfrac12 + \tfrac12\chi_{\bI}(z_\tau) \chi_{\bI'}(x))] \\
        = \tfrac14 + \tfrac14 \E[\chi_{\bI}(x) \chi_{\bI'}(y_\tau)]
                          + \tfrac14 \E[\chi_{\bI}(z_\tau) \chi_{\bI'}(x))]
                          + \tfrac14 \E[\chi_{\bI}(x+z_\tau) \chi_{\bI'}(x+y_\tau)].
    \end{multline*}
    Using the reasoning from \Cref{lem:key-light}, we have
    \[
        \E[\chi_{\bI}(x) \chi_{\bI'}(y_\tau)] \leq \exp(-(|x|+|y_\tau|) \Delta/n) \leq \exp(-4),
    \]
    where we used $|x| \geq \frac{4}{\Delta} n$.
    Similarly $\E[\chi_{\bI}(z_\tau) \chi_{\bI'}(x))] \leq \exp(-4)$, and also
    \[
        \E[\chi_{\bI}(x+z_\tau) \chi_{\bI'}(x+y_\tau)]
        \leq \exp(-(|x+z_\tau| + |x+y_\tau|)\Delta/n)
        \leq \exp(-4),
    \]
    where we used $|x+y_\tau|, |x+z_\tau| \geq \frac{2}{\Delta}n$.
    We may therefore conclude
    \begin{equation}    \label[ineq]{ineq:73}
         \Pr[\text{at least one of $a$'s checks on $(x,y_\tau)$ and $(z_\tau,x)$ is \emph{unsatisfied}}] \geq 1 - (\tfrac14 + \tfrac34 \exp(-4)) \geq .73.
    \end{equation}

    As a consequence, when the (at least~$.98 m$) parity checks in $(T_\tau)_{\tau \in \Upsilon}$  are chosen at random, the number of violations among those imposed on the $(x,y_\tau)$'s and the $(z_\tau,x)$'s stochastically dominates a $\text{Binomial}(.98m, .73)$ random variable.
    A short calculation%
    \footnote{
    For a $\text{Binomial}(t,p)$ random variable $\mathbf X$,
    we know \mbox{$\Pr[ \tfrac 1 t \mathbf X < p - \epsilon] < \exp(- t D(p-\epsilon\| p))$} for~$0 < \epsilon < p$,
    where $D(x\|y) = x \ln(x/y) + (1-x)\ln((1-x)/(1-y))$.
    }
    shows that such a binomial random variable is smaller than $.01 \cdot \frac43 m = .01 n$ with probability at most $O(2^{-1.2 n})$.
    Thus the theorem holds by a union bound, since $2^{(1+o(1))n} \cdot O(2^{-1.2 n}) = 2^{-\Omega(n)}$.
\end{proof}

\subsection{Low-weight words violate many checks}

Recall that $\kappa_{\sch}$ is the second largest eigenvalue of the undirected graph $\sch$.
\begin{assumption}  \label{asmp:kapp}
    We assume $\kappa_{\sch} \leq .00002$.
\end{assumption}
 We recall the following form of the Expander Mixing Lemma~\citep[][Lem.~4.15]{AC88,Vad12}:
\begin{named}{Expander Mixing Lemma}
    If $A_1$ and $A_2$ are subsets of vertices in~$\sch$, and $(\bu,\bv)$ is a uniformly random edge (with orientation), then
    \[
        \abs*{\Pr[\bu \in A_1 \textnormal{ and } \bv \in A_2] - \alpha_1 \alpha_2 } \leq \kappa_{\sch}\sqrt{\alpha_1(1-\alpha_1) \alpha_2(1-\alpha_2)},
    \]
    where $\alpha_i$ denotes $|A_i|/|U|$ for $i = 1, 2$.
\end{named}

Our goal for the remainder of this section is to prove the following:
\begin{theorem}                                     \label{thm:ltc}
    Suppose $\pB$ satisfies the conclusions of \Cref{lem:key-light} and \Cref{thm:key-heavy}.
    Then every $w \in \F_2^{U \times [n]}$ of relative Hamming weight at most $\eps \coloneqq .0002/\Delta$ (i.e., with $|w| \leq \eps \el n$) violates at least $.004 |w|$ parity checks in~$\pB(\vec{\sch})$.
\end{theorem}
\newcommand{\light}{L}
\newcommand{\med}{M}
\newcommand{\heavy}{H}
\newcommand{\viol}{\mathrm{Viol}}
\begin{proof}
Let $\pB$ and $w$ be as given, and write $w = (w_u)_{u \in U}$ for $w_u \in \F_2^n$.
For a given $u \in U$, the subword $w_u$ participates in $2m$ parity checks from $\pB(\vec{\sch})$; we will define $\viol(u)$ to be the number of these checks that are violated.
By double-counting, we need to show that
\begin{equation} \label[ineq]{ineq:iwantu}
    \sum_{u \in U} \viol(u) \geq .008 |w| = .008  \sum_{u \in U} |w_u|.
\end{equation}
Let us say that $u \in U$ is
\[
        \textit{light} \text{ if $|w_u|/n \leq 2/\Delta$,} \quad
        \textit{heavy}  \text{ if $|w_u|/n \geq 4/\Delta$,} \quad
        \textit{medium}  \text{ if $2/\Delta < |w_u|/n < 4/\Delta$.}
\]
We write $\light \subseteq U$ (respectively $\heavy$, $\med$) for the subset of light (respectively heavy, medium) vertices, and we write $\theta_L = |\light|/|U|$ (respectively $\theta_H$, $\theta_M)$ for their fractional size.
By Markov's inequality we have
\begin{equation}    \label[ineq]{ineq:marko}
    \theta_H \leq \eps\Delta/4 \leq .00005 \qquad \text{and} \qquad \theta_L \geq 1 - \eps \Delta/2 \geq .9999.
\end{equation}

Say that $u \in \light \cup \med$ is ``bad'' if $\Pr_{\bv \sim u}[\bv \in \heavy] \geq .9$, where $\bv \sim u$ denotes that $\bv$~is a uniformly random neighbor of~$u$ in~$\sch$.
By the Expander Mixing Lemma, if $\bu \sim \light \cup \med$ is chosen uniformly at random, and then $\bv \sim \bu$ is a randomly chosen neighbor, the probability of $\bv \in \heavy$ is at most
\[
    \theta_H + \kappa_{\sch} \theta_H(1-\theta_H) \leq 1.1 \theta_H
\]
where we used \Cref{asmp:kapp} (with much room to spare).
It thus follows from Markov's inequality that at most a $\frac{1.1\theta_H}{.9} \leq 1.25\theta_H$ fraction of $u \in \light \cup \med$ are bad.
Now for a \emph{good} $u \in \light \cup \med$, which has at least a $.1$-fraction of its neighbors also in $\light \cup \med$, the conclusion of \Cref{lem:key-light} tells us that
\[
    \viol(u) \geq .1 \cdot 2\kreg \cdot .01 \tfrac{\Delta}{\kreg} |w_u| = .002 \Delta |w_u|.
\]
Thus
\begin{equation}   \label[ineq]{ineq:mainsub}
    \sum_{u \in \light \cup \med} \viol(u) \quad\geq\quad  \underbrace{.002 \Delta \sum_{u \in \light \cup \med} |w_u|}_{\textbf{Main}}
    \quad-\quad .01|\heavy|n,
\end{equation}
where the justification for the subtracted term is that
\[
    \sum_{\text{bad } u \in \light \cup \med} .002 \Delta |w_u| \leq  1.25\theta_H|\light \cup \med| \cdot .002 \Delta (4/\Delta)n \leq .01|\heavy| n.
\]
We now divide into two cases.

\paragraph{Case 1:} $|\heavy| \leq .001 |\med|$.
In this case, the subtracted term in \cref{ineq:mainsub} is at most $ .00001|\med|n$, whereas
\[
    \textbf{Main} \geq .002\Delta \cdot |\med| (2/\Delta) n = .004 |\med|n.
\]
Thus the subtracted term is at most a $\frac{.00001}{.004} \le .003$ fraction of $\textbf{Main}$, and hence
the right-hand side of \cref{ineq:mainsub} is at least $.997\cdot \textbf{Main}$.
Furthermore,
\[
    \sum_{u \in \heavy} |w_u| \leq |\heavy| n \leq .001 |\med| n \leq .25 \cdot .004 |\med|n \leq .25 \cdot \textbf{Main}.
\]
Subtracting this from \cref{ineq:mainsub} we conclude that
\[
    \sum_{u \in \light \cup \med} \viol(u) - \sum_{u \in \heavy} |w_u| \geq (.997-.25) \cdot \textbf{Main} \geq .001 \Delta \sum_{u \in \light \cup \med} |w_u|,
\]
which certainly implies \cref{ineq:iwantu} (as we may assume $\Delta \geq 8$).

\paragraph{Case 2:} $|\heavy| > .001 |\med|$.
In this case, by the Expander Mixing Lemma, if $\bu \sim \heavy$ is chosen uniformly at random, and $\bv$ is a random neighbor of~$\bu$ in~$\sch$, the probability of $\bv \in \light$ is at least
\[
    \theta_L - \kappa_{\sch} \sqrt{(1-\theta_H)\theta_L(1-\theta_L)/\theta_H} \geq .9999 - \kappa_\sch \sqrt{(\theta_M + \theta_H)/\theta_H} \geq .9999 - \kappa_\sch \sqrt{1001} \geq .999,
\]
where the first inequality used \cref{ineq:marko}, then we used $(\theta_M + \theta_H)/\theta_H \leq 1001$ from Case~2, and finally the last inequality used \Cref{asmp:kapp}.
Now say that vertex $u \in \heavy$ is ``bad'' if $\Pr_{\bv \sim u}[\bv \in \light] < .99$.
A Markov argument now implies that at most a $\frac{1-.999}{1-.99} = .1$  fraction of $u \in \heavy$ are bad.
As for the at least $.9|\heavy|$ ``good'' $u \in H$, each has~$\kreg$ out-neighbors and~$\kreg$ in-neighbors in~$\vec{\sch}$; by the definition of ``goodness'', for at least a $.98$-fraction of the $\tau \in [\kreg]$ we have that both the $\tau$th out-neighbor \emph{and} the $\tau$th in-neighbor are in~$\light$.
The conclusion of \Cref{thm:key-heavy} then tells us that $\viol(u) \geq .01n$ for each good $u \in \heavy$.
Thus
\[
    \sum_{u \in \heavy} \viol(u) \geq .9|\heavy| \cdot .01 n = .009 |\heavy| n.
\]
Adding this to $.1$ times \cref{ineq:mainsub} yields
\[
    .1 \sum_{u \in \light \cup \med} \viol(u) + \sum_{u \in \heavy} \viol(u)
    \geq .0002 \Delta \sum_{u \in \light \cup \med} |w_u| + .008|\heavy| n
    \geq .008 \sum_{u \in U} |w_u|
\]
(as we may assume $\Delta \geq 40$), and this implies  \cref{ineq:iwantu} as needed.
\end{proof}

\section{Decoding}
\label{decodingsection}
\newcommand{\ezero}{e_0}

In these sections, we discuss decoding the code constructed for \Cref{mainth}.  The code of \Cref{coro} is obtained from this code by a weight-reduction operation followed by a distance-balancing operation.  We show in \Cref{wrc} that efficient decoders for the code constructed in \Cref{mainth} imply efficient decoders for the weight reduced code, up to polylogarithmic factors in the distance to which we can decode.
To decode the distance-balanced code, we can rely on the result of \cite{EKZ20} which gave a general means to decode a distance-balanced quantum code.

\subsection{Decoding cohomology}
In this section we consider the problem of decoding against $X$ errors.  Recall that we have associated $Z$ stabilizers with the $2$-cells of $\cE$.  So, the question of decoding against $X$ errors means reconstructing an error pattern $e$ (up to stabilizers of the code) on $1$-cells from its coboundary $s$ on $2$-cells, where $s$ is referred to as the \emph{error syndrome}, in coding theory terminology.
To express this in a language that is independent of our arbitrary choice to associate $Z$ stabilizers with $2$-cells, we refer to this as decoding cohomology.
We assume throughout that $\nF = \mF = \Theta(\nB) = \Theta(\mB)$.
Our fiber bundle code has $N = \nB \cdot \mF + \mB \cdot \nF = \Theta(\nB^2)$ qubits.

We give a polynomial time decoder 
(more precisely, $\poly(N)$ time) and
show, under the same assumptions as in \Cref{cohomweight},
that it decodes arbitrary errors up to a polylogarithmic fraction of~$d_X$,
meaning that if some error pattern $\ezero$ occurs,
with 
\begin{align}
|\ezero| \le \frac{1}{20} \frac{\mB}{\XConst \Delta^2} \label[ineq]{ezero-bound}
\end{align}
 which is a sufficiently small polylogarithmic fraction of $d_X = \Omega(\mF/\Delta)$,
then the algorithm decodes correctly, computing $\ezero$ up to stabilizers.
Without loss of generality, we may assume that $\ezero$ is a minimal weight chain with the given coboundary.
We also recall the ``shadow'' terminology from \Cref{def:shadow}.

\paragraph{Constructing $e_{\mathrm{arb}}$.}
Given some syndrome $s=\partial^\transp \ezero$,
the first stage of the algorithm is devoted to constructing some arbitrary $1$-chain $e_{\mathrm{arb}}$ satisfying
\[
s=\partial^\transp e_{\mathrm{arb}} \quad \text{ and } \quad |e_{\mathrm{arb}}|_{\mathrm{vsw}} \leq  |s|.
\]
In aid of this, we first define a linear map~$K$ from $p$-chains of $\cE$ to $(p-1)$-chains of $\cB$ (for arbitrary~$p$) by specifying  $K(b^{p-1} \otimes f^1)=b^{p-1}$ and $K(b^p \otimes f^0)=0$.  This map is in a sense ``dual" to the bundle projection map~$\Pi$, as $K$ is nonvanishing on $(p,1)$-cells while $\Pi$ is nonvanishing on $(p,0)$-cells.
In words, $K$ retains the base cells where the chain has an odd number of fiber $1$-cells.
Note that $K(s) = \partial^\transp K(\ezero)$.

The main challenge in constructing $e_{\mathrm{arb}}$ will be constructing a base $0$-chain $e_b$ such that
\[
\partial^\transp e_b = K(s)\quad \text{ and } \quad |e_b| \leq |s|. 
\]
Having constructed this $e_b$, we may fix any $1$-cell $f^1$ in $\cF$, and then \mbox{$\partial^\transp(e_b \otimes f^1)=s$} up to a coboundary of some horizontal chain that can easily be constructed.
Further, $|e_{\mathrm{arb}}|_{\mathrm{vsw}} = |e_b|$.
So it remains to construct~$e_b$.

Note that if we didn't require any bound on~$|e_b|$, obtaining a solution to $\partial^\transp e_b = K(s)$ would be a simple matter of linear algebra: the affine subspace defined by $\partial^\transp e_b = K(s)$ is nonempty (since $\partial^\transp K(\ezero) = K(s)$) and we would just have to find any solution.
To efficiently get the desired low-weight solution, we can run a Sipser--Spielman (belief propagation) type of decoding algorithm (see~\cite[Thm.~10]{SS96}, \cite[Thm.~12.9]{HLW06}).
The setup is slightly nonstandard, in that the roles of the ``code bits'' and ``check bits'' are reversed:
We think of the ``code'' as all solutions of $\partial^\transp e_b = K(s)$, with each base $0$-cell~$j$ enforcing the ``affine parity check'' $\sum_{a \in \partial j} (e_b)_a = 0/1$, the right-hand side depending on whether $j \in K(s)$.
The algorithm begins with the ``faulty'' solution $e_b = 0$, which has Hamming distance at most $|K(\ezero)|$
from a true solution, and repeatedly seeks to toggle a $0$-cell in~$e_b$ so as to decrease the number of violated affine parity checks.
Using the fact (\Cref{prop:cohom-unique}) that we have expansion factor exceeding $\frac34 \Delta$ from the $0$-cells to the $1$-cells in the base (in fact, using \Cref{cor:counique}), the Sipser--Spielman argument directly shows that, so long as $|K(\ezero)| + |K(s)| \leq \frac{1}{\XConst\Delta} \mB$, the algorithm will terminate after at most $|K(s)|$ toggles, yielding an $e_b$ with $\partial^\transp e_b = K(s)$ and $|e_b| \leq |K(s)| \leq |s|$.
Since $|K(\ezero)| \leq |\ezero|$ and $|K(s)| = |\partial^\transp K(\ezero)| \leq 1.01\Delta |\ezero|$, 
we see the algorithm will succeed provided $|\ezero|$ is 
bounded as in \cref{ezero-bound}.

\paragraph{Amending $e_\mathrm{arb}$.}
The second step of the algorithm is again a greedy decoder.
We know by \Cref{prop:cohorep} there exists a bundle $0$-chain $w$ and a base $1$-chain $x$ such that
\[
e_0 = e_\mathrm{arb} + \bdry^\transp w + x \otimes F_0 .
\]
For \emph{arbitrary} $w, x$, let the
{\it horizontal weight} denote the Hamming weight of the horizontal
part of $e = e_{\mathrm{arb}}+\partial^\transp w+ x \otimes F_0$.
We initialize $w=0$.
We greedily choose $x$ to minimize the horizontal weight;
this is linear time, given~$w$.
We then search for a \emph{fixable} base $0$-cell $a$:
\begin{definition}
For an arbitrary bundle $1$-chain $e$, a base $0$-cell $a$ is \emph{amended} if all of the following are true;
otherwise, a base $0$-cell $a$ is \emph{fixable} for $e$.
\begin{enumerate}
	\item[(i)] 	$e$ contains at most $\mF/2$ horizontal cells over any base $1$-cell $b \in \bdry^\transp a$.
	\item[(ii)] 	$e$ contains at most half of all horizontal cells of $\bdry^\transp (a \otimes f^0)$ for any fiber $0$-cell $f^0$.
	\item[(iii)]	Let $A = (\bdry^\transp a) \otimes F_0$. For any $1$-cocycle $z \subseteq A + a \otimes F_1$, where $a \otimes F_1$ consists of all vertical cells over $a$, it holds that $|e \cap A| \le  .8 |(e + z) \cap A| + .2 |A|$.
\end{enumerate}
\end{definition}
\noindent
Remark that the vertical weight or the shadow weight of the vertical part of $e$ has nothing to do with this amendableness.
The condition~(iii) means that $e$ has horizontal weight on $A$ which is fairly close to the minimum possible.
Indeed, all three conditions are trivial if $e$ has the minimum weight in $A$ 
upon adding cocycles ``near $a$'' (i.e., supported on $A + a \otimes F_1$).
The constant $.8$ is chosen because it is close to the approximation ratio achievable by the Goemans-Williamson algorithm.
We will elaborate on this shortly.

Our algorithm will repeatedly find any fixable $0$-cell~$a$ 
and amend it by changing $w$ and $x$ in $e_\mathrm{arb} + \bdry^\transp w + x \otimes F_0$
until no further fixable $0$-cells can be found.
At this point, the algorithm terminates and declares the final chain to be the decoding.
Note that the algorithm terminates after at most~$N$ fixes, 
since the horizontal weight always decreases.

\paragraph{Testing if $a$ is fixable.}

We have to explain how to test whether a base $0$-cell $a$ is fixable.
A simple method is by bruteforce optimization.
As remarked above, all we have to do is to optimize the horizontal weight by adding various cocycles.
This can be done in quasipolynomial time $2^{O(\Delta)}\poly(\mF)$ as follows.
There are $2^{O(\Delta)}$ base $1$-chains supported on $\bdry^\transp a$.
For each such $1$-chain $x$, we consider $e + \bdry^\transp (a \otimes y) + x \otimes F_0$ with a fiber $0$-chain $y$ varying.
Since we only care about the horizontal weight, each bit (a fiber $0$-cell) of $y$ can be independently optimized
by a ``majority vote'' over the horizontal cells in its coboundary.

The bruteforce optimization above is more than necessary.
In fact, our definition of fixable cells is designed to adopt approximate optimization.
Let us first consider the condition~(iii).
Suppose a cocycle $z = z^\mathrm{opt}$ supported on $A + a \otimes F_1$ minimizes $|(e+z) \cap A|$.
Then, the condition~(iii) is equivalent to demanding that
\[
|A| - |A \cap e| \ge .8 \big(|A| - |A \cap (e + z^\mathrm{opt})|\big).
\]
Let us call $|A| - |A \cap e'|$ the \emph{local satisfaction} of $e'$ for any bundle $1$-chain $e'$.
Let us think of a base $1$-chain $x$ supported on $\bdry^\transp a$ 
as a collection of binary variables $x_1,\ldots, x_i,\ldots, x_{|\bdry^\transp a|}$.
Similarly, a fiber $0$-chain $y$ is a collection of binary variables $y_1,\ldots,y_j,\ldots,y_\mF$.
Then a horizontal cell $x_i \otimes \varphi(x_i, a)^{-1} y_j \in \bdry^\transp (a \otimes y_j)$ in $A$ can be given coordinates $(i,j)$.
This horizontal cell at $(i,j)$ over $\bdry^\transp a$
is unoccupied in $e + \bdry^\transp( a \otimes y) + x \otimes F_0$
if and only if $x_i + y_j = e_{i,j} \bmod 2$ where $e_{i,j}$ is the occupancy of~$e$ on~$(i,j)$.
Thus, the local satisfaction is the number of these satisfied equations.
The derandomized Goemans-Williamson algorithm~\cite{MR99} gives an approximately optimal solution $x,y$
such that the local satisfaction of $e + \bdry^\transp(a \otimes y) + x \otimes F_0$
is higher than $.878$ times the optimum local satisfaction.

Once the condition~(iii) is met, we can alternatingly optimize $x$ or $y$ while withholding the other.
Each round of this optimization takes time $\poly(\mF,\Delta)$.
This alternation terminates before $O(\Delta \mF)$ rounds because the local satisfaction must increase.
The conditions~(i) and~(ii) are then fulfilled.
Overall, it takes $\poly(N)$ time to test if $a$ is fixable using the described approximate optimization.

\paragraph{Overview of the analysis.}
Let $w(\tau),x(\tau)$ denote the states of $w,x$ after $\tau$ steps of the algorithm.
Let 
\[
e(\tau)=e_{\mathrm{arb}}+\partial^\transp w(\tau)+ x(\tau) \otimes F_0.
\]
The proof of correctness will take two parts.  In the first part, we assume that at the end of the algorithm the shadow weight of $w(\tau)$ is sufficiently small.  Under this assumption, we show that
when the algorithm terminates it has correctly decoded; i.e., the final $e(\tau)$ is equal to $\ezero$ up to a coboundary.\footnote{Note that even with a maximum likelihood decoder it would not be possible to decode $\ezero$ exactly in general since there may be different error patterns of the same weight and with the same coboundary.}
In the second part we show that indeed the shadow weight of $w(\tau)$ indeed remains sufficiently small throughout the algorithm, using the fact that $|e_{\mathrm{arb}}|_\mathrm{vsw} \leq |s|$.

We will find something unusual in the second part of the proof:
the proof that the shadow weight of $w(\tau)$ remains small 
does not just use the fact that the algorithm terminates after a certain number of steps, 
but rather uses graph-theoretic expansion properties.
Indeed, nothing we show rules out the algorithm running for $\gg \nB$ steps.

\paragraph{First part of the analysis.}
Here we show correctness assuming $|w(\tau)|_{\mathrm{sw}}$ is small.
Note that this implies $|e(\tau)|_{\mathrm{vsw}}$ is also small.


\begin{lemma}
Suppose that the algorithm terminates after $\tau$ steps, and suppose that $|e(\tau)|_{\mathrm{vsw}} + |e_0| \le \tfrac{1}{\XConst \Delta} \mF$.
Then $e(\tau)$ is cohomologous to $\ezero$.
\end{lemma}
\begin{proof}
    Say a $1$-chain $g$ in $\calE$ is a \emph{background} if $\partial^\transp g = \partial^\transp e(\tau)$.
    As per the discussion in \Cref{cohomweight}, for any background~$g$ there is a unique decomposition
    \[
        e(\tau) = g + \partial^\transp w + x \otimes F_0,
    \]
    where the $0$-chain $w =\sum_a a \otimes y_a$ has $|y_a| < \nF/2$ for all~$a$ and where $x$ is a $1$-chain in~$\calB$.
    We call~$w$ the \emph{stabilizer} associated to~$g$, and we call~$x$ the
    \emph{logical}.
    It will also be important for us to keep track of the \emph{stabilizer shadow}~$S$ of~$g$ (i.e., $S$ is the shadow of~$w$).

    The proof will consider a sequence $g_0, g_1, g_2, \dots$ of backgrounds, with the initial background $g_0$ being~$\ezero$.
    We will establish the following properties:
    \begin{enumerate}
        \item \label{itm1} \emph{The logical~$x_t$ of each $g_t$ will be the same for all~$t$.}
        \item \label{itm2} \emph{The stabilizer shadow~$S_t$ decreases in size by~$1$ at each step, until it becomes empty.}
        \item \label{itm3} \emph{The logical $x_t$ vanishes outside the coboundary of~$S_t$, where $\cobd S_t = \bigcup_{a \in S_t} \cobd a$. }
        \item \label{itm4} \emph{The horizontal weight of~$g_t$ on $\partial^\transp S_t$ is smaller than $\mF/100$.}
    \end{enumerate}
    Using only \Cref{itm1,itm2,itm3}, we observe that when the sequence ends at some~$t$, we have an empty~$S_t$, the logical $x_t$ must therefore be everywhere zero, and hence the logical $x_0$ of~$\ezero$ also vanishes.
    This implies that $\ezero$ is cohomologous to~$e(\tau)$, as claimed.

    Let us handle the base case of~$g_0$, and then describe how $g_1, g_2, \dots$ are inductively formed and why \Cref{itm1,itm2,itm3,itm4} hold.
    \paragraph{Base case.}
    Regarding the base case $g_0 = \ezero$, we need to establish \Cref{itm3,itm4}.
    The latter is immediate since the total horizontal weight of~$\ezero$ is already small by the beginning assumption of the decoder.
    As for \Cref{itm3}, from $e(\tau) = \ezero + \partial^\transp w_0 + x_0 \otimes F_0$ we see that on every base $1$-cell $b$ outside of~$\partial^\transp S_0$, we have agreement between $e(\tau)$ and~$\ezero$ up to the addition of~$b \otimes F_0$, this addition being governed by whether~$b$ is in~$x_0$.
    But $\ezero$ has horizontal weight less than~$\nF/2$ on every base $1$-cell, including~$b$.
    Thus $x_0$ must  vanish on~$b$ or else $e(\tau)$ would have horizontal weight more than~$\nF/2$ there, contradicting the condition~(i) of amended cells (any cell in the neighborhood of~$b$ would be trivially fixable).
    This establishes \Cref{itm3}, and thus the base case.

    \paragraph{The inductive construction.}
    Suppose we have formed $g_t$.
    If the stabilizer shadow $S_t$ is empty, we are done.
    Otherwise, we apply \Cref{cor:counique} to $S_t$, obtaining some $a^\bullet \in S_{t}$.
    (Note that $|S_0| \leq \frac{1}{\XConst\Delta}\mF$ by the hypotheses of the lemma, and so all subsequent stabilizer shadows also satisfy this bound, by \Cref{itm2}.)
    Writing $w_t = \sum_a a \otimes y_a$, we form the next background by taking $g_{t+1} = g_t + \partial^\transp (a^\bullet \otimes y_{a^\bullet})$ and $w_{t+1} = \sum_{a \neq a^\bullet} a \otimes y_a$; that is, we simply ``shift'' the $a^\bullet \otimes y_{a^\bullet}$ part of $w_t$ to the background.
    Now \Cref{itm1,itm2} clearly hold, and it remains to verify \Cref{itm3,itm4}.

    Write the neighborhood of $a^\bullet$ in~$B$ as $C \cup D$, where the cells~$C$ are counique neighbors of~$S_t$ and the cells~$D$ are non-counique; \Cref{cor:counique} tells us that $|D| \leq \tfrac14 |C|$.

    \paragraph{Inductively verifying \Cref{itm3}.}
    This is equivalent to showing that $x_{t+1} = x_t$ vanishes on~$C$.
    To do this, we begin by observing that when the decoding algorithm terminated with $e(\tau)$, the cell $a^\bullet$ was not fixable.
    Let us evaluate this fact in the context of the decomposition $e(\tau) = g_t + \partial^\transp w_t + x_t \otimes F_0$.
    Fixing cell $a^\bullet$ amounts to arbitrarily altering $y_{a^\bullet}$ in this decomposition, and then optimizing~$x_t$ to achieve minimal horizontal weight.

    A first observation is that while we always have $|y_{a^\bullet}| \leq \mF/2$, we claim that unfixability of~$a^\bullet$ implies that in fact  $|y_{a^\bullet}| \leq .4\mF$.
    Otherwise, $\partial^\transp w_t$ puts horizontal weight between $.4\mF$ and $.5\mF$ on each cell of~$C$, and from \Cref{itm4} we know that $g_t$ modifies this by at most~$.01\mF$ (collectively, even).
    Thus even after optimizing $x_t$, the contribution to $e(\tau)$'s horizontal weight from~$C$ is at least $.39 \nF |C|$.
    On the other hand, if $y_{a^\bullet}$ were fixed to~$0$, the contribution to $e(\tau)$'s horizontal weight from $C \cup D$ would be at most  $.5 \mF |D| \leq .125 \mF |C|$ from~$D$, and at most $.01\mF$ from~$C$ (the possible contribution from~$g_t$, recalling \Cref{itm3}), for a total of at most $.135 \mF |C|$.
    Then, we have $.8(.135 \mF |C|) + .2\mF(|C| + |D|) \le .358 \mF |C| < .39 \mF |C|$.
    This means $a^\bullet$ \emph{was} fixable by violating the condition~(iii) of the definition of amended cells, a contradiction.
    Thus we have established the claim $|y_{a^\bullet}| \leq .4 \nF$.

    But now we deduce that $g_t + \partial^\transp w_t$ has horizontal weight at most $.41 \mF$ on every cell of~$C$.
    This indeed implies that $x_t$ vanishes on~$C$, since $e(\tau)$ also has horizontal weight less than $\mF/2$ on each cell of~$C$ (indeed, as mentioned earlier $e(\tau)$ has horizontal weight less than $\mF/2$ on every base $1$-cell, else it would be trivially fixable).

   \paragraph{Inductively verifying \Cref{itm4}.}
    Let $H_{t+1}$ denote the horizontal weight of $g_{t+1}$ on $\partial^\transp S_{t+1}$, and $H_t$ the horizontal weight of $g_t$ on $\partial^\transp S_{t}$.
    We will in fact show $H_{t+1} \leq H_t$, which is sufficient to inductively verify \Cref{itm4}.
    We may write  $H_{t+1} - H_t = \mathrm{NEW} - \mathrm{LOSS}$, where $\mathrm{LOSS}$ equals the horizontal weight of~$g_t$ on~$C$, and where $\mathrm{NEW}$ is the weight gain in~$D$ when $g_{t}$ is replaced by~$g_{t+1}$.
    In turn, we can write $\mathrm{LOSS}$ as the sum of contributions $\mathrm{LOSS}_u$ from each fiber vertex~$u$, and similarly write $\mathrm{NEW}$ as a sum of contributions $\mathrm{NEW}_u$.
    Our goal will be to show
    \begin{equation}    \label[ineq]{ineq:loss-new}
        \mathrm{NEW}_u \leq \mathrm{LOSS}_u \quad \forall u.
    \end{equation}

    For a fixed fiber vertex~$u$, consider the bit value of the chain $y_{a^\bullet}$ on~$u$, call it $(y_{a^\bullet})_u$.
    Since $a^\bullet$ is not fixable for $e(\tau)$, this value must be locally optimal (in terms of minimizing horizontal weight) given $g_t$, $x_t$, and given $w_{t+1}$, i.e., given the $y_a$ for $a\neq a^\bullet$, by satisfying the condition~(ii) of the definition of amended cells.
    That is, this value $(y_{a^\bullet})_u$ equals the ``majority vote'' --- across all $j \in C \cup D$ --- of the bits $z_{u,j} \coloneqq (g_t + \partial^\transp w_{t+1}+ x_t \otimes F_0)_{\varphi(j,a^\bullet)^{-1} u}$.
    Recall we already established that $x_t$ vanishes on~$C$ and by construction $C$ is not in the coboundary of the shadow of $w_{t+1}$; thus for $j \in C$ we simply have $z_{u,j} = (g_t)_{\varphi(j,a^\bullet)^{-1} u}$.
    Hence we precisely have $\mathrm{LOSS}_u = |\{j \in C : z_{u,j} = 1\}|$.

    As for $\mathrm{NEW}_u$, it is
    zero if $(y_{a^\bullet})_u = 0$, in which case \Cref{ineq:loss-new} certainly holds.
Suppose instead $(y_{a^\bullet})_u = 1$.
    Then, since the majority vote of~$z_{u,j}$ across $j \in C \cup D$ is~$1$, we must have
    \begin{align*}
        |\{j \in C : z_{u,j} = 1\}| + |\{j \in D : z_{u,j} = 1\}| &\geq (|C|+|D|)/2 \\
        \implies\quad \mathrm{LOSS}_u + |D| &\geq (|C|+|D|)/2 \\
        \implies\quad \mathrm{LOSS}_u  &\geq (|C|-|D|)/2,
    \end{align*}
    so, since $|C|\geq 3 |D|$ we have $\mathrm{LOSS}_u\geq |D|$.
        At the same time, trivially $\mathrm{NEW}_u \leq |D|$,
    verifying \Cref{ineq:loss-new}.
\end{proof}

\paragraph{Second part of the analysis.}
It remains to show that for any step~$\tau$ of the decoding algorithm, $|w(\tau)|_{\mathrm{sw}}$ is sufficiently small compared to $\nB/\Delta$.
Define $Q$ to be the shadow of the horizontal part of $e_{\mathrm{arb}}$,
i.e., $Q$ is the set of all base $1$-cells over which there is some horizontal cell of $e_\mathrm{arb}$.
From the construction of $e_\mathrm{arb}$ we have 
\begin{align}
|Q| \le |\cobd e_b| \le 1.01 \Delta |e_b| \le 1.01 \Delta |\cobd e_0| \le 1.01 \Delta(1.01 \Delta + 2)|e_0| \le 2\Delta^2|e_0| \le \tfrac{1}{10} \tfrac{1}{\XConst}\mB \label{eq:sizeQ}
\end{align}
Define $P(\tau)$ to be the set of all base $0$-cells that were amended in the first $\tau$ steps of the algorithm.
We have that the shadow of $w(\tau)$ is contained in $P(\tau)$;
we may not have equality as it is possible that a cell could be amended multiple times.
We are going to bound $|P(\tau)|$.
Note that $P(0) = \emptyset$.

Suppose a base $0$-cell $a$ is first included in $P(\tau)$ at step $\tau' \le \tau$, i.e., $a \in P(\tau') \setminus P(\tau'-1)$,
then since \mbox{$e(0), e(1), \dots, e(\tau'-1)$} have no support on base $1$-cells outside $Q \cup \cobd P(\tau'-1)$,
at least half of $\cobd a$ must be in \mbox{$Q \cup \cobd P(\tau'-1)$} for~$a$ to be fixable.%
\footnote{
In fact, if more than half of $\cobd a$ are outside the shadow $S'$ of the horizontal part of $e(\tau'-1)$, then $a$ is not fixable.
The reason is as follows.
Any amendment is the change from $e(\tau'-1)$ to $e(\tau'-1) + \cobd (a \otimes y') + x' \otimes F_0$ 
for some fiber $0$-chain $y'$ and some base $1$-chain $x' \subseteq \cobd a$,
where we may assume that $|y'| < \half \mF$ without loss of generality.
Then, $x'$ must vanish outside $S'$ to fulfill the amendedness condition~(i).
This implies that $y'$ must vanish everywhere to fulfill the amendedness condition~(ii).
In turn, this implies that $x'$ must vanish on $S'$
because the amendedness condition~(i) has been and should be obeyed.
}
Hence, every cell $a \in P(\tau)$ must have at least half the cells in its coboundary
either in $Q$ or in the coboundary of some other cell in~$P(\tau)$.
In other words, the number of counique neighbors of~$a \in P(\tau)$ that are not in~$Q$ is at most~$\half |\cobd a|$.
Therefore, the total number of counique neighbors of~$P(\tau)$ that are not in~$Q$ is at most~$\half 1.01 \Delta |P(\tau)| \le .6 \Delta |P(\tau)|$.

On the other hand, if $|P(\tau)| \leq \frac{1}{\XConst\Delta} \mB$,
the number of noncounique neighbors of $P(\tau)$ cannot exceed $.1 \Delta |P(\tau)|$ by \Cref{cor:counique}.
By counting counique and noncounique neighbors of $P(\tau)$ that are not in $Q$, we have
\begin{align*}
	|\partial^\transp P(\tau)| 
		\le |Q| + |\partial^\transp P(\tau) \setminus Q| \le	 |Q| + .1 \Delta |P(\tau)| + .6 \Delta |P(\tau)|.
\end{align*}
But \Cref{prop:cohom-unique} implies that so long as $|P(\tau)| \leq \frac{1}{\XConst\Delta} \mB$,
we know the left-hand side $|\cobd P(\tau)|$ is at least $.9 \Delta |P(\tau)|$; we deduce that
\[
    |P(\tau)| \leq \tfrac{1}{\XConst\Delta} \mB \quad \implies \quad |P(\tau)| \leq \tfrac{1}{.2\Delta}|Q|.
\]
Using $|Q| \le \tfrac 1 {10} \tfrac{1}\XConst \mB$ from \cref{eq:sizeQ},
we see for all sufficiently large $\mB$ that
\[
|P(\tau)| \le \half \tfrac{\mB}{\XConst\Delta} 
\implies |P(\tau+1)| \le |P(\tau)|+1 \le \tfrac{\mB}{\XConst\Delta}
\implies |P(\tau+1)| \le \tfrac{1}{.2\Delta}|Q| \le \half \tfrac{\mB}{\XConst \Delta} .
\]
We conclude that $|P(\tau)|$ is always bounded by $\half \tfrac{1}{\XConst \Delta}\mB$.

\subsection{Decoding homology}
We now give a proposed algorithm for decoding homology.  We conjecture, but do not prove, that this algorithm decodes errors
of weight up to a polylogarithmic fraction of $d_Z$.

 The algorithm takes as input a $0$-chain $s_0$ in $\cE$ which contains the syndrome.  It initializes a $1$-chain $u$ in $\cE$ to $0$.  It initializes some $0$-chain $s$ in $\cE$ to $s_0$.  As the algorithm proceeds, it modifies the chain $u$ by a sequence of local updates explained below.  After each update, the algorithm then updates $s$ so that $s=s_0+\bdry u$.
 The algorithm attempts by these local updates to reduce $|s|$.

Define a {\it fiber string} of length at most $r$ to be a vertical $1$-chain of Hamming weight at most $r$ whose boundary consists of exactly $2$ $0$-cells.  The string can be thought of as stretching between between these two $0$-cells.  Here we assume that $r<\ell<n_F$.

The algorithm has a counter $r$, initialized at $r=0$.  The algorithm loops over $r=0$ to $r=\ell/\polylog(n_F)$.
For update $r$, the algorithm performs a greedy search, trying to
reduce $|s|$ by either adding to $u$
a fiber string of length at most $r$ or by adding to $u$ some horizontal cell plus some sum of fiber strings of length at most $r$.

The algorithm performs this greedy search for the given $r$ until it is not possible to reduce $|s|$ further, at which point it increments $r$ and continues the loop if $r<\ell/\polylog(n_F)$.

After the loop terminates, the algorithm then takes the given error chain $s$ and bundle projects it to the base, giving a $0$-chain in $\cB$.  It then attempts to find some $1$-chain $w$ in $\cB$ such that $\partial w$ is equal to the given $0$-chain in $\cB$ and so that $|w|$ is small compared to the distance of the base code.  (We conjecture that the base code is such that, if the weight of the bundle projected $s$ is small enough, then such a $w$ can be found by a greedy search; this seems easier to prove than some other properties we need also.)

Finally, given $w$, one may find some horizontal $1$-chain $x$ in $\cE$ whose bundle projection is equal to $w$, and with $|x|=|w|$; for example, one may simply take $x=w \otimes f^0$ for any fiber $0$-cell $f^0$.
Then, the chain $u+x$ has the property that its boundary, after bundle projection, is equal to $s_0$ after bundle projection, and so one can add some vertical $1$-chain $y$ to $u+x$ to obtain a $1$-chain whose boundary is $s_0$;
The algorithm then outputs this $1$-chain.

Let us sketch why we conjecture this works.
Suppose the true error pattern was some $1$-chain $e$ with $\bdry e=s_0$.
It is possible that $e$ might, for example, even be a sum of stabilizers, in which case perhaps $s_0=0$ even though $e\neq 0$.
Without loss of generality, we may assume that $e$ is a minimal weight error pattern with $\bdry e=s_0$.
After the loop terminates, we have that $e+u$ is given by some sum $h+v$ of horizontal and vertical chains.
We expect that $|h|$ will still be small compared to $d_Z$ and so $|\partial h|$ will be large, indeed proportional to $|h|$ in some way depending on $\Delta$.
We expect however that many of the $0$-cells in $h$ will be attached to fiber strings of length $>r$ in $v$, where ``attached" means that one of the two cells in the boundary of that string is the given $0$-cell in $h$.
However, assuming the greedy algorithm does not increase the weight of $v$ too much compared to the weight of vertical cells in $e$, then for $r=\ell/\polylog(n_B)$, the weight of $h$ must then be small compared to $n_B$.
So, at this point, the algorithm has (assuming these conjectures are correct) computed $e$ up to some $1$-chain $h+v$ with $|h|$ small.
The algorithm then returns some other horizontal chain $x$, with $|x|$ small,
and so computes $e$ up to $h+v+x+y$, with $h,x$ horizontal and having small Hamming weight and $v,y$ vertical, i.e., it computes $e$ up to a closed $1$-chain such that the Hamming weight of this chain on horizontal cells is small compared to the distance of the base code.
However, any such closed $1$-chain is homologically trivial.

\subsection{Decoding cohomology against erasure errors in almost linear time}

\newcommand{\erased}{{\mathsf D}} 
\newcommand{\certain}{{\mathsf C}} 
\newcommand{\peel}{{\mathsf P}} 
\newcommand{\findpeel}{{\mathfrak P}}

With \emph{erasure errors},
we are given  syndrome bits and a specific set $\erased$ of qubits (erased qubits)
on which there are potential errors.%
\footnote{
One may consider an equivalent setting
in which the erased qubits are simply lost with no errors on other qubits;
in this case, one may initialize those qubits arbitrarily and then measure stabilizers.
}
The location of the true errors is unknown
and it is the decoder's goal to determine the true errors up to stabilizers.
Here we present an algorithm for this erasure decoding problem
where the actual errors are $X$, denoted as a bundle $1$-chain $x \in \calE_1$,
under the assumption that~$|\erased|$ is less than a polylogarithmic fraction of~$\mF$,
which is smaller than~$d_X$ by \Cref{cohomweight}.

If we find any chain $x'$ that reproduces the given syndrome, i.e., $\bdry^\transp x' = \bdry^\transp x$,
with the constraint that $x' \subseteq \erased$,
then the combination $x + x'$ of the actual errors and a correction
is coclosed and has weight less than $d_X$,
and hence is a coboundary.
It is thus obvious that this erasure decoding problem is solved in time $|\erased|^3 \poly(\Delta)$
since $x'$ is a solution of an inhomogeneous system of linear equations of $|\erased|$ variables
that participate in $|\erased| \poly(\Delta)$ equations.
It is however not so obvious whether such $x'$ can be found in linear time in $|\erased|$.

\begin{lemma}\label{lem:coho-erasure-decoding}
Assume the supposition of \Cref{cohomweight}.
If $|\erased| < \mF / (\XConst \Delta^2)$,
then in time $|\erased| \poly(\Delta,\log |\erased|)$
we can find a chain $x'$ such that $x' + x$ represents the zero cohomology class.
\end{lemma}

\begin{proof}
The algorithm is based on belief propagation that deforms $\erased$ to decrease $|\erased|$ down to zero.
Let us define two sets $\certain \subset \erased$ and $\peel \subset \calE_{0,0} \times \calE_{0,1}$.
Roughly speaking, $\certain$ consists of cells on which we know errors with certainty,
and $\peel$ is the set of cells of $\erased$ that we can ``peel off'' to produce more cells that will qualify for $\certain$.
We will consider the time complexity later.

\paragraph{Removing obvious errors.}
$\certain$ consists of all cells $e \in \erased$
such that some $2$-cell touches no other cells of $\erased$ but $e$.
The error on $e$ is thus determined.
By removing $\certain$ from $\erased$ and recording the determined errors
and iterating these,
we can sometimes eliminate all of $\erased$.
For example, if $\erased(0)$ at time step~0
is a collection of consecutive horizontal cells over a single base $1$-cell,
occupying an interval in the fiber,
then $\certain(0)$ would be just two end cells
and the shrunk $\erased(1) = \erased(0) - \certain(0)$ will have two end cells in $\certain(1)$.
This continues until $\erased(t)$ becomes empty.
If this iteration does not eliminate all of $\erased$,
then we are left with $\erased(t_1) \neq \emptyset$
where every 2-cell on the coboundary of some cell of $\erased(t_1)$
meets at least two $1$-cells of $\erased(t_1)$.
Let $\erased(t_j)$ be any configuration such that $\certain(t_j) = \emptyset$.

\paragraph{Representing an error on one cell by others using stabilizers.}
We use $X$-stabilizers (associated with bundle $0$-cells) to rescue the situation.
An $X$-stabilizer is~$\bdry^\transp u$ for some bundle $0$-cell $u$,
and $\bdry^\transp u$ contains exactly two vertical cells,
``between'' which there are $(1 \pm .01)\Delta$ horizontal cells.
Suppose that (i) $\erased(t_j)$ contains at least one of the the vertical cells, say~$v$, of~$\bdry^\transp u$,
and that (ii) $\erased(t_j)$ contains
more than half of the horizontal cells of $\bdry^\transp u$
which would belong to $\certain(t_j)$ if $v \in \bdry^\transp u$ were absent from $\erased(t_j)$.
Given $\erased(t_j)$,
we collect all pairs $(u,v)$ satisfying (i),(ii) to form~$\peel(j)$.
For some $(u,v) \in \peel(j)$,
we set $\erased(t_j + \half) = (\erased(t_j) - v) \cup (\bdry^\transp u - v)$.
We will show that~$\peel(j)$ is nonempty shortly using \Cref{prop:cohom-unique},
but let us see first why this is a rescue.

It is important to note that any error (no error or $X$) on $v$
can be expressed by errors on $\bdry^\transp u - v$ up to $X$-stabilizers (coboundaries of $0$-cells).
So, for any correction on $\erased(t_j)$ there is an equivalent correction on $\erased(t_j + \half)$.
Seeking a correction on $\erased(t_j + \half)$, rather than on $\erased(t_j)$,
potentially increases overall weight of the correction,
but the increment is not big as we show:
$\erased(t_j + \half)$ lacks~$v$, so the condition (ii) for $\peel(j)$ implies that $\certain(t_j +\half)$
has more than half of the horizontal cells of $\bdry^\transp u$,
i.e., $|\certain(t_j + \half)| > \half( |\bdry^\transp u| - 2)$.
Since $\erased(t_j)$ contains more than half of the horizontal cells of $\bdry^\transp u$ and also $v$,
we see $|\erased(t_j + \half)| < |\erased(t_j)| + \half (|\bdry^\transp u| - 2)$.
Therefore, removing $\certain(t_j + \half)$ from $\erased(t_j + \half)$,
we obtain $\erased(t_j + 1)$ whose cardinality is strictly less than $|\erased(t_j)|$.

Hence, if we employ $\peel(j)$ whenever $\certain(t_j) = \emptyset$,
we can always decrease $|\erased(t)|$ until $\erased(t)$ becomes empty.
The final correction $x'$ has weight bounded by $ |\erased|(1+0.6 \Delta)$
because each transition $\erased(t_j) \to \erased(t_j + \half)$
can enlarge the support of $x'$ by the number of added cells which is less than $1 + \half 1.01 \Delta$.
It follows that $x+x'$ has weight less than $|\erased| \Delta$ which is less than $d_X$,
implying that $x+x'$ represents the zero cohomology class.

\paragraph{$\peel(j) \neq \emptyset$ whenever $\erased(t_j) \neq \emptyset$ but $\certain(t_j) = \emptyset$.}
Consider the \emph{shadow} of the vertical cells of $\erased(t_j)$ onto the base.
Here by shadow we mean as before the set of all base $0$-cells
above which there is at least one vertical cell in $\erased(t_j)$.
Since the cardinality of the shadow is less than $m/(\XConst\Delta)$
(which follows by induction in $j$ with the assumption that $\Delta |\erased(0)| < m / (\XConst\Delta)$),
\Cref{prop:cohom-unique} implies that
there is a base $0$-cell $a$ such that more than half of its coboundary $1$-cells are counique neighbors.
Let us go to the fibers over $a$ and its counique neighbors $b_1, \ldots, b_q$. 
Any vertical cell $v = a \otimes f^1 \in \erased(t_j)$ above $a$
must share every $2$-cell $b \otimes f^1$ in its coboundary with some other $1$-cells of $\erased(t_j)$.
If $b$ here is a counique neighbor of $a$,
then the 2-cell $b \otimes f^1$ cannot be shared with any other vertical cell of $\erased(t_j)$,
so it must be shared with a horizontal cell of $\erased(t_j)$.
Conversely, if any horizontal cell over a counique neighbor $b$ of $a$
shares a 2-cell with a vertical cell in their coboundary,
it must do with a vertical cell over $a$.
Hence, within $\erased(t_j)$
all the vertical cells over $a$ and all the horizontal cells over the counique neighbors of $a$,
together form a graph $\findpeel$ whose every node is linked to some other node by sharing a $2$-cell.

Consider the \emph{spine} of $\findpeel$;
the spine is the union of the images of all nodes of $\findpeel$
under the map $a \otimes f^1 \mapsto a \otimes f^1$
and $b \otimes \varphi(b,a)^{-1} f^0 \mapsto a \otimes f^0$ for any counique neighbor $b \in \bdry^\transp a$.
The spine has (far) less than $\nF$ elements,
and hence divided into consecutive clusters.
Any cluster cannot have a $0$-cell (that comes from a horizontal cell of $\findpeel$) at either end;
if it did, there would be a horizontal cell that is exposed to a bundle $2$-cell
that only sees this horizontal cell.
Let $v = a \otimes f^1$ be the $1$-cell (that comes from a vertical cell of $\findpeel$) at the bottom end of a cluster.
Any horizontal cell linked to $v$ in $\findpeel$ must be $b \otimes f^0_\mathrm{top}$ for some counique neighbor $b \in \bdry^\transp a$
where $f^0_\mathrm{top}$ is the top end of $f^1$ ($\bdry f^1 = f^0_\mathrm{top} + f^0_\mathrm{bottom}$)
because $v$ is at the bottom of a cluster in the spine.
Therefore, $\erased(t_j)$ contains the horizontal cells
$b_i \otimes \varphi(b_i,a)^{-1} f^0_\mathrm{top}$ but not $b_i \otimes \varphi(b_i,a)^{-1} f^0_\mathrm{bottom}$ for all $i = 1, \ldots, q$.
If we deleted $v$ from $\erased(t_j)$, then all $b_i \otimes \varphi(b_i,a)^{-1} f^0_\mathrm{top}$ would belong to $\certain(t_j)$.
Now, a bundle $0$-cell $u = a \otimes f^0$ has coboundary
that contains $v = a \otimes f^1$ and all $b_i \otimes \varphi(b_i,a)^{-1} f^0_\mathrm{top}$ for $i = 1,\ldots,q$.
Since $q$ is more than half the degree of $a$,
we see that $(u,v) \in \peel(j)$.

\paragraph{Time complexity.}
One should not compute $\certain$ and $\peel$ every time $\erased$ is updated.
Instead, they should be initially computed once by going over all cells of $\erased$ and small neighborhoods,
and every time a cell is removed or added to $\erased$,
the sets $\certain$ and $\peel$ should be updated accordingly.
For each cell of $\erased$ it takes time $O(\Delta)$ to determine its membership to
$\certain$ and $\peel$.
When removing a cell of $\certain$ from $\erased$, there are $O(\Delta^2)$ cells to examine,
so it takes time $\tilde O(\Delta^3)$ to remove a cell and update the sets.
In the transition $\erased(t) \to \erased(t + \half)$,
we alter $O(\Delta)$ cells so it takes $\tilde O(\Delta^4)$ to complete this transition.
The total number of removals is bounded by $|\erased(0)| \Delta$,
yielding overall time complexity $\tilde O(|\erased(0)|\Delta^5)$.
\end{proof}

\section*{Acknowledgments}
MBH thanks Mike Freedman for explaining spectral sequences; they weren't needed since the homology was computed in a more elementary way, but they helped in understanding the homology of fiber bundles.
RO thanks Venkat Guruswami for coding theory discussions, and thanks Microsoft Quantum for hosting him throughout the time this work was completed.
We thank G. Z\'{e}mor for pointing out an error in the weight reduction result of \cite{Has17qic}.

\bibliographystyle{alpha}
\bibliography{quantum}

\newpage
\appendix
\section{Weight Reduction and Chain Homotopy}
\label{wrc}
In this appendix, we show how to weight-reduce the classical base code, and then construct the fiber bundle code as a bundle over that.
We will spend a fair amount of time developing some generalities.  One of the main messages is that a way to think of two codes as being ``equivalent" is the notion of homotopy equivalence of chain complexes.  This is an old notion in homological algebra; we will combine it with an additional requirement of some ``Lipschitz bounds" on certain maps.
Without these Lipschitz bounds, the notion of equivalence would be somewhat vacuous for quantum codes as any two codes with the same number of logical qubits (and, if any, the same number of redundant $X$- and $Z$-stabilizer generators) would be equivalent.
These Lipschitz bounds will enable us to give a sharper notion of equivalence, relating the distance of the codes, and proving that efficient decoding for one code implies efficient decoding for an equivalent code up to some fraction of the distance.

Then, having given these generalities, we will show first a homotopy equivalence of two different classical base codes, which is then used to show a homotopy equivalence of the fiber bundle code to a weight-reduced code.

\subsection{Review on Chain maps and homotopies}
For two chain complexes $\cA,\cB$,
a {\it chain map} $f$ from $\cA$ to $\cB$ is a linear map from $\cA_j$ to $\cB_j$ for each ~$j$, such that
\[
\partial_\cB f = f \partial_\cA.
\]
Given two chain maps $f,g$ from $\cA$ to $\cB$, a {\it chain homotopy} is a linear map
$h$ from $\cA_j$ to $\cB_{j+1}$ for each $j$ such that
\[
f-g=h \partial_\cA + \partial_\cB h.
\]
In this case, the maps $f,g$ are said to be {\it homotopic}.

These definitions of chain map and chain homotopy are both standard definitions, long established in the literature.
Intuitively, one may think of a chain map between chain complexes as being an analogue of a continuous function between topological spaces.

Note that any chain map $f$ induces a map on homology, $f_* : H_j(\cA)\rightarrow H_j(\cB)$, since it maps cycles to cycles.
For the same reason, the transpose of $f$, written $f^\transp$, induces a map on cohomology $f^* : H^j(\cB) \to H^j(\cA)$.
Of course, these maps need be neither injective nor surjective; 
for example the map $f$ that maps every chain to the zero chain is a chain map.
Our interest later will be in particular chains maps where the induced map is both injective and surjective.

The following lemma is standard:
\begin{lemma}
\label{samemaphom}
Any two homotopic chain maps induce the same maps on homology and on cohomology.
\begin{proof}
We have
$f-g=h \partial_\cA + \partial_\cB h,$
and $h\partial_\cA$ vanishes on cycles and $\partial_\cB h$ is a boundary.
For cohomology, take the dual (transpose) of everything: we have $f^\transp-g^\transp=\partial_\cA^\transp h^\transp+h^\transp \partial_\cB^\transp$, and
$h^\transp \partial_\cB^\transp$ vanishes on cocycles while $\partial_\cA^\transp h^\transp$ is a coboundary.
\end{proof}
\end{lemma}
Remark: the transpose of a chain map commutes with the coboundary operator; such a map is commonly called a cochain map.

To motivate this definition of a chain map, one might have instead considered the following alternative definition inspired by the concept of a homotopy between two functions (this is also standard; see for example \cite{NCL}).
Let $\cI$ be a chain complex which is a cellulation of the interval $[0,1]$; for example, we may choose $\cI$ to have one $1$-cell labelled $e$ and two $0$-cells, labelled $v_0$ and $v_1$, such that $\partial e=v_1-v_0$.
Then, one might define a homotopy $m$ to be a chain map from $\cI \otimes \cA$ to $\cB$, such that
for any $\cA$-chain $u$ we have $m(v_0 \otimes u)= f(u) $ and $m(v_1 \otimes u )= g(u)$.

However, let us see that this new definition is the same as the old.
Since $m$ is a chain map,
$$\partial (m(e \otimes u))= m((v_1-v_0) \otimes u )-
 m(e \otimes(\partial u)).$$
Define
$h(w)=-m(e \otimes w)$, so the above equation becomes
$-\partial h(u)=g(u)-f(u)+h(\partial u).$
So, we have recovered a chain homotopy $h$ from $m$.

Similarly, suppose we have a chain homotopy
$h$ such that $f-g=h \partial_\cA + \partial_\cB h$.
Define a linear map~$m$ from $\cI \otimes \cA$ to $\cB$ as follows.
Let $m(v_0 \otimes u)=f(u) $ and $m(v_1 \otimes u )=g(u)$ as before
and define $$m(e \otimes u)= -h(u).$$
Then
one may verify that $m$ is a chain map.
So, the two definitions are indeed equivalent.

Given two chain complexes $\cA,\cB$, a {\it homotopy equivalence} 
is a pair of chain maps $
\xymatrix{
\cA \ar@/^/[r]^{f} & \cB \ar@/^/[l]^{g}
}
$ such that
$\cA \xrightarrow{f} \cB \xrightarrow{g} \cA$ is homotopic to the identity map on $\cA$ 
and $\cB \xrightarrow{g} \cA \xrightarrow{f} \cB$ is homotopy to the identity map on $\cB$.

The following is again well known:
\begin{lemma}
\label{homiso}
A homotopy equivalence induces isomorphisms on homology and cohomology.
\end{lemma}
\begin{proof}
If $f : \cA \to \cB$ and $g: \cB \to \cA$ give a homotopy equivalence,
$f_* \circ g_*$ is the identity on the homology of $\cB$ by \cref{samemaphom},
and $g_* \circ f_*$ is the identity on the homology of $\cA$;
that is, $(f_*)^{-1} = g_*$.
The same is true for $f^*$ and $g^*$.
\end{proof}

\subsection{Distance and Decoding}
We now consider the effect of these chain maps on the weight of a chain.  Given a chain $u$, write $|u|$ to denote the {\it weight} of the chain.  In all examples in this paper, the weight will denote the Hamming weight (number of nonzero coefficients of the chain in the given basis), but the next definition and lemma work for any choice of weight function.

We say that a chain map $f:\cA \to \cB$ is {\it Lipschitz} with Liptschitz constants $K_j = K_j(f)$ 
if $|f(u)| \leq K_j |u|$ for any $j$-chain $u$.
Let $d_j(\cA)$ denote the minimum weight achieved by a representative of nontrivial homology class of $H_j(\cA)$ 
and let $d^j(\cA)$ denote the minimum weight achieved by a representative of nontrivial cohomology class of $H^j(\cA$).

\begin{lemma}
Let
$
\xymatrix{
\cA \ar@/^/[r]^{f} & \cB \ar@/^/[l]^{g}
}
$
be a homotopy equivalence with $f,g^\transp$ Lipschitz.
Then
\[
    d_j(\cA)\geq K_j(f)^{-1} d_j(\cB), \qquad \text{and} \qquad
    d^j(\cA) \geq K_j(g^\transp)^{-1} d^j(\cB).
\]
\begin{proof}
If $u$ is a nontrivial representative of $H_j(\cA)$, then $f(u)$ is a nontrivial representative of $H_j(\cB)$ and so $d_j(\cB) \leq |f(u)|\leq K_j |u|$.
The second inequality is similar.
\end{proof}
\end{lemma}
\noindent
We will use this lemma later to lower-bound the distance of the weight-reduced code that we construct.  Our bounds will almost certainly not be tight: all we will use is a certain Lipschitz constant in a map.  We will remark further on this later.

Given a decoder for a code defined by $\cB$,
we can decode errors on a code defined by $\cA$
if the chain complexes $\cA$ and $\cB$ are homotopically equivalent as we now explain.
As we always discuss quantum CSS codes
where $j$-cells are qubits, $(j+1)$-cells are $Z$-stabilizers, and $(j-1)$-cells are $X$-stabilizers,
the following discussion will be completely symmetric for homology and cohomology,
so we only explain homology decoding, i.e., correcting $Z$ errors from observed violations of $X$-stabilizers.
For the code defined by $\cA$, a $Z$ error pattern is some unknown $j$-chain $e_\cA$,
and the syndrome is $s_\cA = \partial e_\cA$.
Similarly, $e_\cB$ stands for $Z$ errors on the code defined by $\cB$ and $s_\cB = \partial e_\cB$ is the syndrome.

Then:
\begin{lemma}
Suppose $g \circ f$ is homotopic to the identity on $\cA$ by $h: \cA_{j-1} \to \cA_{j}$.
Assume there is some decoding algorithm $O$ for the code associated with $\cB$ such that, given a syndrome $s_\cB=\partial e_\cB$ as input,
it computes some $\tilde e_\cB$
such that $s_\cB=\partial \tilde e_\cB$, where $e_\cB$ and $\tilde e_\cB$ are
homologous whenever $|e_\cB|\leq W(O)$, for some bound $W(O)$.
Then, given some $s_\cA = \partial e_\cA$ as input, 
the following algorithm computes a chain $\tilde e_\cA$ such that
$s_\cA=\partial \tilde e_\cA$, where $e_\cA$ and
$\tilde e_\cA$ are homologous whenever $|e_\cA|\leq K_j(f)^{-1}\cdot W(O)$:

\begin{itemize}
\item[({\bf 1})] Call  $O$ with $f(s_\cA)$ as syndrome.
The result will be some error pattern $\tilde e_\cB$ such that
$f(s_\cA) = \partial \tilde e_\cB$.

\item[({\bf 2})] Return
\[
\tilde e_\cA = g(\tilde e_\cB) - h(s_\cA).
\]
\end{itemize}

\begin{proof}
We have
\begin{align*}
&&f(s_\cA)&= \partial \tilde e_\cB \\
&\implies& 
g(f(s_\cA)) & =  g(\partial \tilde e_\cB)\\
&\implies& 
g(f(s_\cA)) & =  \partial g( \tilde e_\cB) &\text{(by definition of chain map)} \\
&\implies&  s_\cA+h\partial s_\cA +\partial h s_\cA &= \partial g( \tilde e_\cB) &\text{(by definition of homotopy equivalence)} \\
&\implies& 
s_\cA&= \partial g( \tilde e_\cB)-\partial h s_\cA &\text{(since $\partial s_\cA=0$)}
\\
&\implies& s_\cA &= \partial \tilde e_\cA.
\end{align*}

Let us verify that $\tilde e_\cA$ is homologous to $e_\cA$, assuming that $|e_\cA|\leq K^{-1}_j(f) \cdot W(O)$.
We have
$s_\cA=\partial e_\cA$ so $f(s_\cA)=\partial f(e_\cA)$.
By assumption on the decoding algorithm for the code associated with~$\cB$, the chain
$\tilde e_\cB$ is homologous to $f(e_\cA)$, in which case
$g(\tilde e_\cB)$ is homologous to $g(f(e_\cA))$.
We have
$g(f(e_\cA))=e_\cA+h\partial e_\cA+\partial h e_\cA,$
so $g(f(e_\cA))$ is homologous to $e_\cA+h \partial e_\cA=e_\cA+h s_\cA$.
So, $g(\tilde e_\cB)$ is homologous to $e_\cA+h s_\cA$ implying that
$\tilde e_\cA$ is homologous to $e_\cA$.
\end{proof}
\end{lemma}

\subsection{Cell Combining and Collapsing and Weight-Reducing the Classical Codes}

We now define a procedure of ``cell combining" and a related ``cell collapsing.''
Cell combining can be intuitively understood as combining two cells into a single one, 
such as combining a pair of edges into one edge.
This procedure establishes a homotopy equivalence of two chain complexes.
This will be useful in the next subsections as the original code 
(either the classical code or the fiber bundle code) 
can be derived from a weight-reduced code by the procedure.
Cell collapsing is a ``dual'' of this procedure, explained below.

\subsubsection{Cell combining}
An illustration of cell combining that may be useful is shown in \Cref{fig:cc}.  

\begin{figure}
\centering
\includegraphics[width=3in]{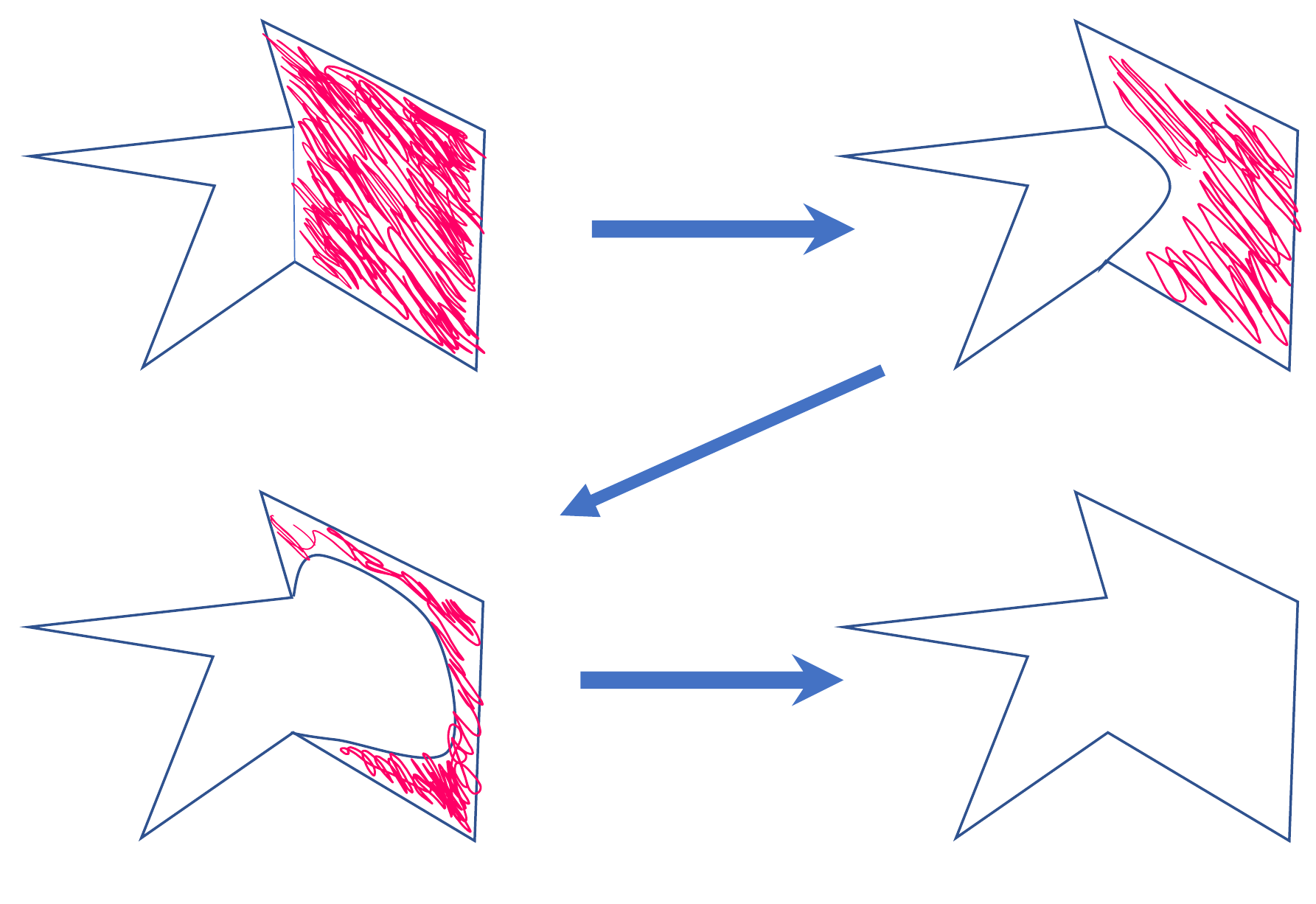}
\caption{
Illustration of cell combining.  The image shows a sequence of three steps.
We begin with two ``$1$-cells", shown at the top left.  
One cell (which corresponds to $e_1$ in the definition) is colored white, 
while the other cell (which corresponds to $e_2$) is colored red.
The straight lines represent ``$0$-cells" in the boundary of these ``$1$-cells."
The straight line between the two cells corresponds to $v$.
We show pictorially a sequence as $v$ is deformed, 
until eventually $e_2$ shrinks to nothing 
and the line which corresponds to $v$ has been mapped to the other four ``$0$-cells" in the boundary of $e_2$.
}
\label{fig:cc}
\end{figure}

\begin{definition}
\label{ccd}
(Cell combining)
Let $\cA$ be some chain complex with only $0$- and $1$-cells.
Let $v$ be some $0$-cell in $\cA$ with only
two $1$-cells in its coboundary, written $e_1,e_2$.
(With signed coefficients, we take a convention that 
$\partial e_1 = v + \cdots$ whereas $\partial e_2 = -v + \cdots$.)
We say that the following cell complex $\cB$ is given by combining cells $e_1,e_2$ in complex $\cA$.

The set of $1$-cells of $\cB$ is the set of $1$-cells of $\cA$ other than $e_1,e_2$ and with an extra cell $e$.
The set of $0$-cells of $\cB$ is given by the set of $0$-cells of $\cA$ other than $v$.  
The boundary operator $\partial_\cB$ is given by
by $\partial_\cB = f \circ \partial_\cA \circ g$, where the maps $f,g$ are as follows.

Let $f : \cA \to \cB$ be such that $e_1 \mapsto e$, $e_2 \mapsto 0$, 
and $f$ maps all other $1$-cells in $\cA$ to the corresponding $1$-cell in $\cB$.
Further, $f$ maps all $0$-cells in $\cA$ except~$v$ to the corresponding $0$-cell in $\cB$, 
while it maps $v$ to the chain which in $\cA$ would be written $v + \partial_\cA e_2$.
Let $g : \cB \to \cA$ be such that $e \mapsto e_1-e_2$, 
and $g$ maps all other $1$-cells to the corresponding $1$-cell in $\cA$
and all $0$-cells to the corresponding $0$-cell in $\cA$.
\end{definition}

\begin{lemma} \label{combinelemma}
The maps $f,g$ of \Cref{ccd} are chain maps, giving a homotopy equivalence of the complexes $\cA,\cB$ in the definition.
\end{lemma}

\begin{proof}
The composition $f \circ g: \cB \to \cB$ is the identity, so tautologically is homotopic to the identity.

Note that $g\circ f$ maps $v$ to $v+\partial e_2$, maps $e_1$ to $e_1+e_2$, maps $e_2$ to $0$, and acts as the identity on all other cells of $\cA$.
Thus, $g \circ \partial_\cB = g f \partial_\cA g$ is equal to $\partial_\cA g$ possibly except for $e$,
but at $e$ we also have $(g f \partial_\cA g)(e) = \partial (e_1 + e_2) = (\partial_\cA g)(e)$.
Hence, $g$ is a chain map.

Let $I$ denote the identity map on $\cA$.
Then $gf - I$ is equal to $0$ on all cells, 
except it maps $v$ to $\partial e_2$, maps $e_1$ to $e_2$, and maps $e_2$ to $-e_2$.
Define $h : \cA \to \cB$ to map $v$ to $e_2$ and to vanish on all other $0$-cells; 
$h$ vanishes on $1$-cells.
Then indeed $gf - I = h \partial_{\cA} + \partial_{\cA} h$ so $gf$ is homotopic to the identity.

To verify that $f$ is a chain map, we need to verify $\partial_\cB f = f \partial_\cA$; that is, we need
$f \partial_{\cA} gf=f  \partial_{\cA}$.  So it suffices to verify that $f \partial_{\cA} (\partial_{\cA} h + h \partial_{\cA})=0$.
Since $\partial^2=0$, we need $f \partial_{\cA} h \partial_{\cA}=0$.  However, $e_2$ is the only cell in the range of $h$ and
$f \partial_{\cA} e_2=0$.
\end{proof}

\subsubsection{Cell collapsing}

Cell collapsing is the dual: given some chain complex $\cA$ with a $1$-cell $e$ having two $0$-cells $v_1,v_2$ in its boundary, we can collapse $e$, mapping it to $0$ by a chain map and mapping $v_1,v_2$ to the same image.  Rather than give this dual map explicitly, simply note that given some chain complex with boundary operator $\partial$ obeying the conditions of \Cref{combinelemma}, if we interchange $1$-cells and $0$-cells, and use boundary operator $\partial^\transp$, then we have such a $1$-cell $e$.  Then, since \Cref{combinelemma} constructs maps $f,g$ which commute with $\partial$, the maps $g^\transp,f^\transp$ commute with $\partial^\transp$.

\subsection{Weight reducing}

We can now explain how to weight-reduce the classical code.
Consider an arbitrary classical code, $\ker(\cA_1 \xrightarrow{\partial} \cA_0)$,
with bits associated with $1$-cells and checks associated with $0$-cells of some chain complex.
Assume that each $1$-cell has at most $d_1$ $0$-cells in its boundary 
and each $0$-cell has at most $d_0$ $1$-cells in its coboundary.
The procedure to reduce $d_1$ can be understood intuitively as making several copies of each bit, 
adding checks that enforce that copies of the bit are the same, 
and having different copies participate in different checks; 
the procedure to reduce $d_0$ can be understood as the dual.

Define a weight-reduced classical code as follows.
For each bit $b$ of the original code,
the weight-reduced code will have $|\partial b|$ bits,
labeled by a pair $(b,c)$ where $c \in \partial b$.
Similarly, for each check $c$ of the original code,
and the weight-reduced code will have $|\partial^\transp c|$ checks,
labeled by a pair $(c,b)$ where $b \in \partial^\transp c$.
We add also $|\partial b|-1$ additional checks for each bit $b$,
labeled by a pair $[b,j]$ for $j\in \{1,\ldots,|\partial b|-1\}$.
Similarly we add $|\partial^\transp c| - 1$ additional bits for each check $c$, 
labeled by a pair $[c,k]$ for $k \in \{1,\ldots,|\partial^\transp c| - 1\}$.

We will call the bits $[c,k]$ and checks $[b,j]$ the ``auxiliary" bits and checks.
As mentioned, these auxiliary bits and checks are in addition to 
the bits $(b,c)$ with $c \in \partial b$ and the checks $(c',b')$ with $b' \in \partial^\transp c'$.
Thus if the original code has $N_1$ bits and $N_0$ checks, 
and has $E$ nonzero elements in the boundary map $\partial$ (i.e., edges in the bipartite Tanner graph),
the weight-reduced code has $O(N_0+E)$ checks and $O(N_1+E)$ bits.

We now specify the boundary operator $\partial_\text{wr}$ of the weight-reduced code.
We continue to use $\partial$ without subscript $\text{wr}$ to mean the boundary operator of the original code.
For each bit $b$ of the original code, 
let us order $\partial b$ once and for all to write $\partial b = c_1 + \cdots + c_{|\partial b|}$.
We let the auxiliary check $[b,j]$ have coboundary $(b,c_j)+(b,c_{j+1})$ for each $j$.
Similarly, 
ordering $\partial^\transp c$ to write $\partial^\transp c = b_1 + \cdots + b_{|\partial^\transp c|}$,
we let the auxiliary bits $[c,k]$ have boundary $(c,b_k)+(c,b_{k+1})$ for each $k$.
Finally, we put $(c,b) \in \partial_\text{wr} (b,c)$ where $c \in \partial b$.
All coefficients of $\partial_\text{wr}$ other than those specified in this paragraph vanish.
In this weight-reduced code, every bit, either $(b,c)$ or $[c,k]$, participates in 2 or 3 checks
and every check, either $(c,b)$ or $[b,j]$, depends on 2 or 3 bits.

We make the following claim:
\begin{lemma}
\label{che}
The weight-reduced code $\cA$ is homotopy equivalent to the original classical code $\cB$
by $
\xymatrix{
\cA \ar@/^/[r]^{f} & \cB \ar@/^/[l]^{g}
}
$.
The maps $f,g,f^\transp,g^\transp$ have Lipschitz constants which are all $O(\max(d_1,d_2))$.
\end{lemma}
\begin{proof}
We repeatedly apply cell combining to remove all the auxiliary checks 
and then cell collapsing to remove all the auxiliary bits.
This gives a homotopy equivalence.

To get Lipschitz constants, consider this in two steps:
first removing auxiliary checks then removing auxiliary bits.
After removing all auxiliary checks, 
the image (under the map $g$ from the original code to the weight-reduced code) of any bit $b$ 
is the sum of bits $(b,c)$ with $c \in \partial b$ and there are at most $d_1$ such bits.
The image under $g$ of any check $c$ is the corresponding check in the weight-reduced code.
So, the Lipschitz constant on $1$-cells is $O(d_1)$ and on $0$-cells is $O(1)$.
Then after cell collapsing, the Lipschitz constant on all cells is $O(\max(d_0,d_1))$.

The calculation of Lipschitz constants for $f$ is similar: 
after the first step, the image of any auxiliary check $[b,j]$ 
is contained in the boundary of $b$ and so has Lipschitz constant $O(d_1)$.
\end{proof}

\subsection{Weight-reducing the fiber bundle code}

We now explain how to weight-reduce the fiber bundle code.
The fiber bundle code was defined by a classical base code as well as a choice of twists.
We will define a weight-reduced fiber bundle code to be the code given by a bundle over the weight-reduced base code, with a certain choice of twists.

The choice of twists is the obvious one: originally, there was a twist $\varphi(b,c)$ between a bit $b$ and a check $c$ if $c \in \partial b$.
In the weight-reduced base code, 
there is a Tanner graph edge between a bit $(b,c)$ and a check $(c,b)$.
We put a possibly nonzero twist $\varphi_\text{wr}((b,c),(c,b)) = \varphi(b,c)$,
and zero twists elsewhere.
In particular, there is zero twist on any element of the boundary operator involving an auxiliary bit or check.

\begin{lemma}
The weight-reduced fiber bundle code is homotopy equivalent to the fiber bundle code defined previously.
The pair of homotopy equivalence maps and their duals have Lipschitz constants which are all $O(\max(d_1,d_2))$.
\begin{proof}
In the untwisted case, the result would follow from \Cref{che}.  This follows from a general result: suppose $\cA,\cB,\cC$ are some chain complexes and $f,g$ give a homotopy equivalence of $\cA,\cB$.  Then the pair of $f\otimes I$ and $g\otimes I$ is a homotopy equivalence between $\cA \otimes \cC$ and $\cB \otimes \cC$.

From this equivalence in the untwisted case, the equivalence in the twisted case follows using the gauge redundancy.  
Recall that gauge redundancy is a chain isomorphism between two bundles over the same base;
it is a special instance of homotopy equivalence of Lipschitz constant~$1$.
Concretely, 
for any base $1$-cell $b$, one can choose an arbitrary fiber automorphism $\varphi_b \in \ZZ_\mF$ and change
the twists as $\varphi(b,a) \to \varphi(b,a) + \varphi_b$ for all $a \in \partial b$.
Likewise, one can choose an arbitrary fiber automorphism $\varphi_a \in \ZZ_\mF$ for any base $0$-cell $a$,
and modify the twists as $\varphi(b,a) \to \varphi(b,a) + \varphi_a$ for all $b \in \cobd a$.
These modifications, referred to as gauge transformations, 
are simply redefinitions of the coordinate system in the fiber.
Clearly, all twists around any given base $1$- or $0$-cell can be made zero by some gauge transformation.

Each step of cell combining (or collapsing) in the equivalence for the classical code involves only a single bit (or check) of the original code.
For any such bit (or any check), we can use the gauge redundancy so that all twists involving that bit (or that check) are equal to zero.  Then, given homotopy equivalences $f,g$ for the combining (or collapsing step), we use equivalence $f \otimes I,g\otimes I$ for the fiber bundle code.  
Since $f,g,h$ act as the identity except for cells associated with the given bits or checks of the original code, 
it suffices to check homotopy equivalence only on those, for which the result follows from the untwisted case.
Thus, the homotopy equivalence is established as claimed.

This homotopy equivalence is a composition of $O(N)$ ``local'' homotopy equivalences,
but the Lipschitz constant is much smaller than the product of those local Lipschitz constants.
In fact, along the sequence of local homotopy equivalences,
the image of one cell is merged with other cells or split into several cells only once.
Hence, the Lipschitz constant of the composition is at most the maximum of those of local homotopy equivalences.
\end{proof}
\end{lemma}

Given these results, it follows that the distance of the fiber bundle code over the weight-reduced classical code is within a polylogarithmic factor of the distance of the fiber bundle code over the original classical code.  This distance bound is 
probably not optimal: it does not use the fact that, for example, a chain on the weight-reduced bundle code representing an element of homology must be closed, but only uses the fact that its image under the map is closed.  Similarly to how we did not optimize logarithms earlier, we do not worry about this here.

\section{Notation}
\label{notationsection}
Here we present a brief list of some notation used in the paper.  The notation is not in alphabetical order but rather is in a rough order of when the concepts are introduced.
\begin{itemize}
\item[$\cB$:]  the base of the bundle

\item[$\calF$:]  the fiber of the bundle

\item[$\cE$:] the fiber bundle.  Qubits of a quantum code are associated with $1$-cells of this bundle while checks of the code are associated with $0$- and $2$-cells.

\item[$\nB$:]  the number of $1$-cells in the base.

\item[$\mB$:]  the number of $0$-cells in the base.

\item[$\nF=\mF$:]  the number of $0$-cells in the fiber which is the same as the number of $1$-cells in the fiber.

\item[$B$:] a bipartite graph defining the base.  The \magenta{right} vertices represent bits (variables) of a code which correspond to $1$-cells of the base, and the \magenta{left} vertices represent checks of a code which correspond to $0$-ells of the base.

\item[$F$:] a cycle graph.

\item[{$[n]$:}] shorthand notation for the set of integers $\{1,2,\ldots,n\}$.

\item[$\ell$:] an integer chosen so that $\nF=\mF=\ell^2$.  All twists are integer multiples of $\ell$.

\item[$\XConst$:] a large universal constant related to bounding expansion in the base graph; see \Cref{prop:cohom-unique}.

\item[$U$:] the set of used twists: $U = [\mF/\ell] = [\ell]$.

\item[$n,m$:] shorthand notation used sometimes for $\nB,\mB$.

\item[$\Delta$:] the average check-degree of the base code.  Ultimately, $\Delta$ is chosen to be $\Theta(\log^2 n)$.

\item[$\kreg$:]  the number of distinct twists.  Ultimately, $k$ is chosen to be $\Theta(\log n)$.

\end{itemize}
\end{document}